\documentclass[11pt,draftclsnofoot,onecolumn]{IEEEtran}
\usepackage[T1]{fontenc}
\usepackage[latin9]{inputenc}
\usepackage{amsthm}
\usepackage{amsmath}
\usepackage{graphicx}
\usepackage{amssymb}
\usepackage{color}
\usepackage[normalem]{ulem}

\makeatletter
\theoremstyle{remark}
\newtheorem{thm}{Theorem}

  \theoremstyle{remark}

  \theoremstyle{remark}
  \newtheorem{cor}{Corollary}

  \theoremstyle{remark}
  \newtheorem{rem}{Remark}

  \theoremstyle{remark}
  \newtheorem{lemma}{Lemma}

  \theoremstyle{remark}

\newcommand{\vn}{\mathbf{n}}
\newcommand{\vx}{\mathbf{x}}
\newcommand{\vy}{\mathbf{y}}

\newcommand{\vl}{\mathbf{L}}

\newcommand{\e}{\mathcal{E}}

\newcommand{\bmu}{\boldsymbol{\mu}}
\newcommand{\bSigma}{\boldsymbol{\Sigma}}

\newcommand{\D}{\mathfrak{D}}
\newcommand{\Dec}{\mathfrak{Dec}}

\newcommand{\co}{\text{I}}
\newcommand{\cc}{\text{II}}
\newcommand{\ccc}{\text{III}}

\newcommand{\Sign}{\text{Sign}}

\newcommand{\argmin}{\operatornamewithlimits{argmin}}
\newcommand{\argmax}{\operatornamewithlimits{argmax}}

\makeatother

\begin{document}

\title{Two-way Decode-and-Forward for Energy-Efficient Wireless Relaying: Selective Forwarding versus One-bit Soft Forwarding}
%\title{Two-Way Decode-and-Forward Protocols for Energy-Efﬁcien Wireless Relaying with Hard Decoding Selective Forwarding versus One-Bit Soft Forwarding}

\author{Qing F. Zhou, Wai Ho Mow, Shengli~Zhang and Dimitris Toumpakaris
\thanks{Qing F. Zhou is with Department of Communication Engineering, School of Computer and Information, Hefei University of Technology, and was with the Department of Electronic and Computer Engineering, Hong Kong University of Science and Technology.}
\thanks{Wai Ho Mow is with the Department of Electronic and Computer Engineering, Hong Kong University of Science and Technology, Clear Water Bay, Hong Kong S.A.R., China. }
\thanks{Shengli~Zhang is with Shenzhen University, P.R.China.}
\thanks{Dimitris Toumpakaris is with the University of Patras, Greece.}
\thanks{This work was mainly supported by the AoE grant E-02/08 from the University Grants Committee of the Hong Kong S.A.R., China; partially supported
by NSFC 61471156, 61372078 and National 973 Project 2013CB336700.} }

\maketitle
\thispagestyle{empty} %  show no page number

\begin{abstract}
Motivated by applications such as battery-operated wireless
sensor networks (WSN), we propose an easy-to-implement
energy-efficient two-way relaying scheme.
In particular, we address the challenge of improving the standard two-way selective decode-and-forward protocol (TW-SDF) in terms of block-error-rate (BLER) with minor additional complexity and energy consumption. By following the principle of soft relaying, our solution is the two-way one-bit soft forwarding (TW-1bSF) protocol in which the relay forwards
the one-bit quantization of a posterior information metric about the transmitted bits, associated with an appropriately designed reliability parameter.

In WSN-related standards (such as IEEE802.15.6 and Bluetooth), block codes are adopted
instead of convolutional and other sophisticated codes, due to their efficient decoder hardware implementation. As the second main contribution, we derive tight upper bounds on the BLER performance for both TW-SDF and TW-1bSF, when the two-way relaying network employs block codes and hard decoding. The error probability analysis confirms the superiority of TW-1bSF. Moreover, we derive the asymptotic performance gain of TW-1bSF over TW-SDF, which further suggests that the proposed protocol is a good choice, especially when long block codes are used.
%Furthermore, we also address the impact of the error-correcting capability of the code on the performance gain.
\end{abstract}

%\newpage

\section{Introduction}

Traditionally, the design of network protocols has mainly focused on how
to maximize the throughput or bandwidth efficiency. Recently, the
design of energy-efficient (a.k.a. green) wireless networks has generated considerable interest.
Instead of the power consumption related to transmission, the term
energy efficiency typically refers to the energy consumed by communication units to process
signals, including encoding and decoding operations.
 For example, in some remote monitoring applications using
wireless sensor networks (WSNs), the sensor nodes are expected to be of
low cost and to have many years of battery life. Therefore, energy
efficiency is the most critical design issue for such networks. Motivated
by such application scenarios, in this paper we propose an easy-to-implement
energy-efficient \emph{two-way} relaying scheme in which two source nodes
exchange their messages with the help of an intermediate relay node, and we evaluate its performance.
%\sout{The three-node two-way relay network, though simplistic, may shed
%some light on the design of a general multi-hop two-way relay network
%and similar network models have been widely considered in the literature.}

Applying the idea of network coding \cite{Li2003} to two-way relaying networks is particularly simple and fruitful. This is because network
coding over the binary field (i.e. XOR) is sufficient for this topology. Moreover, under
the setting of time-division multiple access, exchanging one pair of messages between the two nodes requires
only three packet transmissions, instead of four when using routing.
%\sout{In the literature, this three-transmission scheme alternatively goes for the names of digital network coding scheme \cite{Katti2007}, straightforward network coding scheme \cite{Zhang2006} and non-physical-layer network coding scheme \cite{Liew2011}. Here, we follow the terminology of \cite{Zhang2006} and name this scheme as the \emph{direct network coding} scheme.}
Here, the three-transmission relaying method is referred to as \emph{direct network coding} (DNC). Allowing the source nodes to broadcast simultaneously, as proposed in physical-layer network coding (PNC) \cite{Zhang2006,Liew2011} and analog network coding (ANC) \cite{Katti2007}, further reduces the number of transmissions to two, resulting in even higher spectral efficiency.
%\sout{Due to the promising sepectral efficiency boosting, PNC and ANC have attracted intensive attention and been studied extensively \cite{Hausl2006,Wilson2010,Lang2010,Lang2010a,Zhou2010,Zhou2010a}.}
However, these two spectrally efficient schemes are not amenable for wireless networks with energy efficiency requirements because of the need for accurate synchronization, sophisticated decoding at the relay, no use of the direct link, and accurate channel estimation \cite{Gao2009a,Gao2009b}.
%\sout{For the reasons mentioned above, the direct network coding scheme is more preferable to be adopted in the wireless network which consists of cheap devices and can only afford low power consumption for signal processing.}

There are two variations of the two-way decode-and-forward protocols that are based on DNC, namely, non-selective and selective.
In the non-selective two-way decode-and-forward (DF) protocol,
the relay decodes both packets received from the two sources, and then broadcasts
the network-coded (XOR) version of the two decoded packets regardless
of whether the two decoded packets are erroneous or not. This
effect of error propagation may result in serious performance
degradation. A modification of the DF protocol is the \emph{two-way selective decode-and-forward} (TW-SDF)
protocol. In this protocol, a sufficiently strong CRC code is used to assess
the correctness of the two decoded packets at the relay. The relay only processes the two decoded packets if they have both been decoded correctly and then broadcasts
their network-coded combination. This approach effectively avoids error propagation, and can successfully recover the
loss on the diversity order of the non-selective DF protocol. The decode-and-forward scheme outperforms its demodulation-and-forward counterpart based on joint network-channel coding \cite{Li2010,Ozdemir2012b}.
%Since the SDF protocol is superior
%to the DF protocol, it will be considered as the benchmark in
%our work.

In \emph{one-way} relaying networks, the SDF protocol can be further improved by using the soft relaying technique \cite{Sne:SDF1,Li:DTC-SIR,Bui:SDF2,Karim2010, Zei:Q-Design}.
Specifically, the relay uses a soft decoder to derive the
\emph{a posteriori} probabilities of the code bits, and then forwards the
soft information values. At the desired receiver, the forwarded soft information values are exploited as a priori information by source-controlled soft decoders \cite{Hagenauer1995} to improve decoding performance.
It is shown that performance improvement can be achieved over the SDF protocol with hard decoding, because in many
occasions an incorrectly decoded packet contains only a few
erroneous bits and the \emph{a posteriori} information forwarded
by the relay may help locate the position of the erroneous bits. However,
the main disadvantage of the soft relaying protocol is that
it requires significant additional bandwidth and/or power consumption, mainly due to the requirement of transmitting soft values or their multiple-bit
quantized version and because of the use of complicated soft-input soft-output
decoders. The key motivation behind our work is to simplify the soft relaying method by using only a one-bit quantized
representation of the soft information, and hence avoid the bandwidth expansion associated with soft signal forwarding and the complexity because of soft decoding. To mitigate the error propagation due to the low-rate quantization, the quantized bit message is forwarded along with a reliability parameter, which is utilized by the decoder to estimate the equivalent LLR of the message from the relay.
%\sout{The idea is inspired by our earlier work on
%the use of soft-input hard-output (e.g. Viterbi) decoder in low-complexity
%iterative decoding with one-bit extrinsic information representation.}
The so-constructed \emph{one-bit soft forwarding} (1bSF) protocol is
almost as easy to implement and as energy efficient as the SDF protocol
that employs hard decoding \cite{Li2006} (see also \cite{Krzymien2005}). In \cite{Dai2009} preliminary results regarding the effectiveness of
the protocol for one-way relay networks were obtained through computer
simulation.
In this paper, we extend the philosophy to two-way relaying systems by
proposing the two-way 1bSF (TW-1bSF) protocol.

In the literature, the performance evaluation of various relaying protocols is typically carried out using the outage probability when applying ideal capacity-approaching codes at the relay(s) \cite{Laneman2004,Zhou2009}, or via the symbol error rate (SER) in the absence of channel coding at the relay(s) \cite{Hasna2003,Ribeiro2005,Zhou2010a}. To the best of the authors' knowledge, few works have analytically derived the performance of relaying protocols that use practical channel codes. These include the loose lower bound derivation of the BER for a Turbo-coded one-way relay system \cite{Roy2006}, and the approximate derivation of the block-error-rate (BLER) for a convolutional-coded one-way relay system \cite{Amat2010}, for an end-to-end coded two-hop detect-and-forward one-way relay system \cite{Benjillali2010}, and for a demodulation-and-forward scheme with network coding but no channel coding \cite{Ozdemir2012}. However, from the viewpoint of communication engineering, a tight upper bound is of higher practical importance. To bridge this gap, in this paper we derive tight BLER upper bounds for TW-SDF and TW-1bSF, when block codes are employed and hard decoding is used at the receivers. We further derive the asymptotic performance gain of TW-1bSF
over TW-SDF, and discuss the design of the reliability parameter to improve the performance of the proposed TW-1bSF protocol.
The error probability analysis confirms the superiority of TW-1bSF. It further shows that the proposed protocol is a good choice, especially when long block codes are used.

\section{System Model and Related Work}

\begin{figure}[h]
\center\includegraphics[scale=0.4]{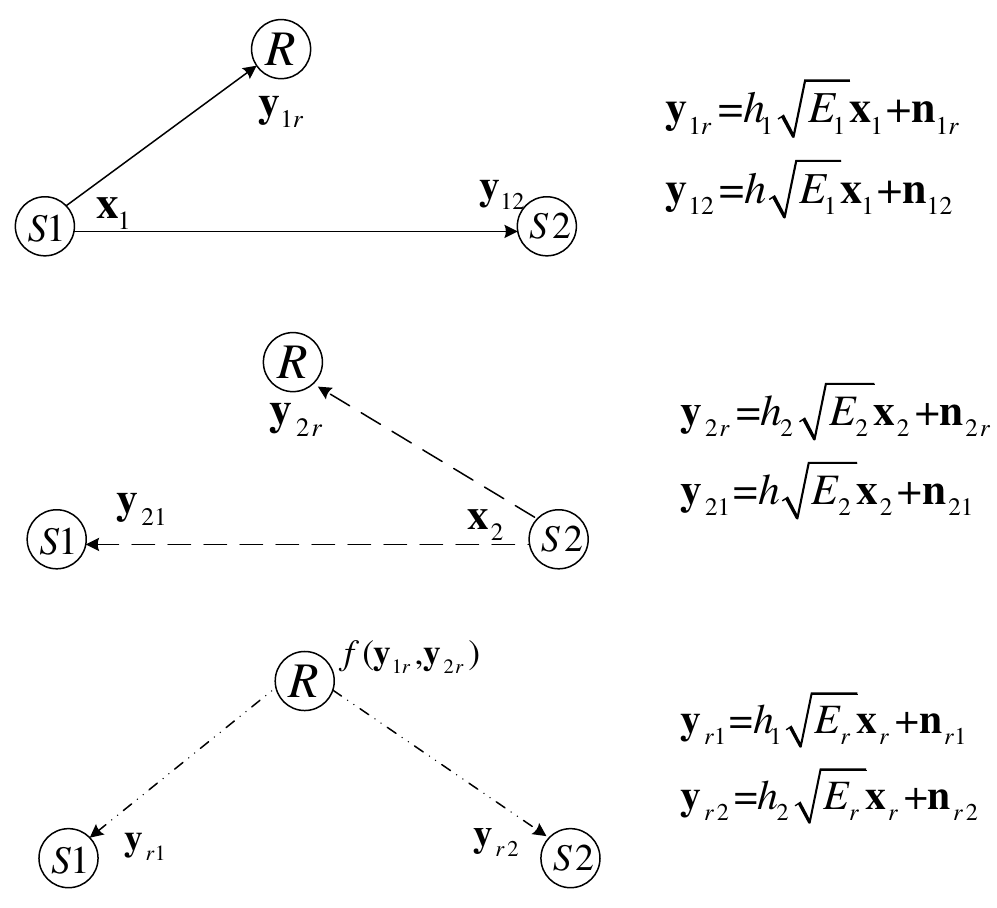}
\caption{\label{fig:system_model}Two-way relaying based on direct network coding (DNC).}
\end{figure}

% ## figure 1
%################
%\begin{figure}[h]
%\center\includegraphics[scale=0.6]{TWR_3ts}
%\caption{\label{fig:system_model}Two-way relaying with direct network coding.}
%\end{figure}
%
%In this part, it is briefly described the system model for two-way
%relaying networks with the use of Network Coding.
We consider a three-node two-way wireless relaying network.
In this network, two sources $S1$ and $S2$
exchange their information messages with the help of a relay $R$. The messages are encoded to protect them against channel impairments.
%Specifically, $S1$ transmits its length $n$ codeword $\vx_{1}=[x_{1}(1),x_{1}(2),\cdots,x_{1}(n)]$ to $S2$, and $S2$ transmits $\vx_{2}=[x_{2}(1),x_{2}(2),\cdots,x_{2}(n)]$
%to $S1$, {\color{blue}where $x_{i}(k)\in\{\pm1\}$ is the $k$th modulated symbol of the codeword $\vx_{i}$. For simple implementation, BPSK, block channel coding and ML hard decoder are considered in this paper. To avoid interference, $S1$
%and $S2$ are assumed to transmit at two consecutive time slots.}
Fig.\ref{fig:system_model} shows the discrete-time model of DNC-based two-way relaying over this network.
During the first time slot, $S1$ broadcasts a vector $\vx_{1}$ containing the encoded message to $R$ and $S2$,
using power $E_1$.
%Denote $\vy_{ij}$ the received signal
%at the node $i\in\{1,2,r\}$ ($r$ for $R$, $1$ for $S1$,
%2 for $S2$) during the $j$th time slot.
The received signal vectors at $R$ and $S2$ are $\vy_{1r}$ and $\vy_{12}$, respectively.
During the second time slot, $S2$ broadcasts $\vx_{2}$ using power
$E_2$. The received signal vectors at $R$ and $S1$ are $\vy_{2r}$ and $\vy_{21}$, respectively.
After receiving the signals $\vy_{1r}$ and $\vy_{2r}$, the relay generates a signal vector
$\vx_r = f(\vy_{1r},\vy_{2r})$ for forwarding, where $f(\cdot,\cdot)$ is a mapping
determined by the relaying scheme. Without loss of generality, it is assumed that the power of $\vx_r$ is
normalized to 1.
%If anyone of $\vy_{12}$ and $\vy_{21}$ is decoded erroneously at its receiver, then,
During the third time slot, $R$ broadcasts $\vx_r$ with power $E_r$.
The signal vectors received at
$S1$ and $S2$ are $\vy_{r1}$ and $\vy_{r2}$, respectively.

In this system model, $\vx_i=[x_{i,1},x_{i,2},\cdots, x_{i,n}]^T\in \mathcal{X}^{n}$ with $i\in \{1,2\}$ and $x_{i,m}\in\{+1,-1\}=\mathcal{X}$, i.e., we assume that BPSK modulation is employed. The same encoder, and therefore the same codebook is used by both sources. In addition, $\vy_{ij}=\{y_{ij,m}\}_{m=1}^n$ represents the received signal vector
at node $j\in\{1,2,r\}$ ($r$ for $R$, $1$ for $S1$,
2 for $S2$) from node $i$; $h$, $h_{1}$ and $h_{2}$ are the channel coefficients of the $S1$-$S2$ link (direct link),
the $S1$-$R$ link and the $S2$-$R$ link, respectively, and all are reciprocal; $\vn_{ij}$ are noise vectors at receiver $j$ with independent identically-distributed (i.i.d.) Gaussian entries $\sim \mathcal{N}(0,\sigma_j^2)$. We assume that the channel coefficients are fixed during the three slots of the two-way relaying transmission scheme. Moreover, as can been seen in Fig.~1, flat fading is assumed. We also assume that all noise terms have the same normalized variance $\sigma_{r}^{2}=\sigma_{1}^{2}=\sigma_{2}^{2}=N_{0}/2 =1/2$,
and that the sources utilize the same power $E_1=E_2=E$.
%Then the transmitted signal-to-noise ratio (SNR) from the source notes is measured by $E/N_{0}$, and the transmitted SNR from $R$ by $E_r/N_{0}$.
%For simplicity, we further assume $N_0=1$.
Note that in this system model, all terms are real-valued because of the use of BPSK modulation.

In the DNC-based two-way relaying method, it is crucial to carefully design the mapping function $ f(\vy_{1r},\vy_{2r})$ such that both messages can be decoded by their desired receivers. It is even more desirable to optimize it to achieve the best performance in terms of capacity or bit error rate. Using the concept of SDF to design $f(\vy_{1r},\vy_{2r})$ leads to the TW-SDF protocol, which is described in the following.

\subsection{Two-way SDF Relaying Protocol (TW-SDF)}
The relay $R$ decodes $\vy_{1r}$ and $\vy_{2r}$ at the end of the first and the second time slot, respectively, to generate $\hat{\vx}_{1r}$ and $\hat{\vx}_{2r}$, i.e., the estimates of $\vx_1$ and $\vx_2$. Then the relay combines $\hat{\vx}_{1r}$ and $\hat{\vx}_{2r}$ using network coding, and forwards the network-coded vector $\vx_r$ to the sources during the third time slot. If any of the $\hat{\vx}_{1r}$ and $\hat{\vx}_{2r}$ is erroneously decoded, the decoding error will propagate from $R$ to the sources, leading to performance degradation.
%When such decoding error occurs, the conventional decode-and-forward
%(DF) protocol still forwards the XOR-encoded codeword. However, the decoding scheme at $S2$ mentioned previously does not work appropriately because the error
%propagated from $R$ would damage, rather than helping, in decoding the signal
%$\vy_{12}$ from the direct link. Due to the error propagation, it is well known
%DF suffers diversity order loss.
%
%Here we assume that any decoding error can be detected by CRC.
To eliminate this detrimental effect, TW-SDF does not allow
$R$ to forward erroneous information by forming the forwarded signal vector as
\begin{equation}\label{eq:xrSDF}
\vx_{r}^{\text{TW-SDF}}=
\begin{cases}
\vx_{1} \circ \vx_{2}, & \text{if }\hat{\vx}_{1r}=\vx_{1}\text{ and }\hat{\vx}_{2r}=\vx_{2} \\
\text{Nil}, & \text{otherwise}
\end{cases}
\end{equation}
where $\circ$ denotes the Hadamard product, which is the element-wise XOR operation in the modulated signal domain\footnote{Let $x_i = (-1)^{b_i}$, $b_i\in\{0,1\}$ represent BPSK modulation. The value of the XOR operation of two bits is given by $(-1)^{b_i\oplus b_j} = (-1)^{b_i}(-1)^{b_j}=x_i x_j$. So, the XOR operation in the BPSK modulated signal domain is $x_r = x_i x_j$.}. In \eqref{eq:xrSDF}, Nil represents no transmission. Note that $\hat{\vx}_{ir}=\vx_{i}$ means that no error is detected by, e.g., an error-detection code.

To demonstrate the decoding algorithm at the source nodes, we focus on the receiving process at $S2$ without loss of generality. After the third time slot, $S2$ has the received signals $\vy_{12}$ and $\vy_{r2}$. While $\vy_{12}$ is affected by Gaussian noise only, $\vy_{r2}$ is distorted by Gaussian noise and the transmitted codeword $\vx_2$ from $S2$. Since $S2$ knows the transmitted codeword of its own, it can remove the contamination caused by $\vx_2$ from $\vy_{r2}$, and generate a new detection statistic
\begin{equation}\label{eq:def-yr}
\tilde{\vy}_{r2}  \triangleq \vy_{r2}\circ\vx_{2}
=  h_{2}\sqrt{E_r}\vx_{r}^{\text{TW-SDF}}\circ\vx_{2}+ \tilde{\vn}_{r2},
\end{equation}
where $\tilde{\vn}_{r2}=\vn_{r2}\circ\vx_2$ is a vector of i.i.d. Gaussian r.v's $\sim \mathcal{N}(0,N_0/2)$. If $\hat{\vx}_{ir}=\vx_{i}$, then $\tilde{\vy}_{r2}=h_{2}\sqrt{E_r} \vx_{1} +\tilde{\vn}_{r2}$. By summing $\vy_{12}$ and $\tilde{\vy}_{r2}$ using maximal ratio combining (MRC), $S2$ can decode the codeword $\vx_1$ as
\begin{equation}\label{eq:decTWSDF}
\hat{\vx}_{1}^{\text{TW-SDF}}=
\begin{cases}
\Dec( 4h\sqrt{E}\vy_{12}+4h_{2}\sqrt{E_r}\tilde{\vy}_{r2}), &  \text{if }\hat{\vx}_{ir}=\vx_{i} \\
\Dec(4h\sqrt{E}\vy_{12}), & \text{otherwise}
\end{cases}
\end{equation}
where $\Dec$ represents a decoding algorithm.

\section{Two-way One-bit Soft Forwarding (TW-1bSF) Protocol}
In this paper, our aim is to improve the TW-SDF protocol by leveraging the soft relaying principle, while using negligible additional bandwidth and by keeping the decoding algorithm simple such that the energy efficiency requirement can be met. The main idea is to always transmit during the third time slot even when a $\hat{\vx}_{ir}$ is in error. Moreover, one indicator bit for the entire forwarded packet is transmitted. As will be shown, this leads to significant performance improvement compared to the TW-SDF protocol. Because the one-bit indicator corresponds to the entire packet of length $n$, the required additional bandwidth is negligible.

\subsection{Design of TW-1bSF}

The \emph{a posteriori} soft output for each message/code bit, typically the LLR, is composed of a sign component and a reliability component.
Using the entries of the received signal vector $\{y_{ir,m}\}_{m=1}^n$, the LLR of $x_{i,m}$ can be calculated as
\begin{equation}\label{eq:LLRdef}
\text{LLR}(x_{i,m}|y_{ir,m}) = 4h_i \sqrt{E} y_{ir,m}.
\end{equation}
In \eqref{eq:LLRdef} we have used the fact that BPSK is employed. In vector form, the LLR of $\vx_{i}$ based on $\vy_{ir}$ is
\begin{equation}
\text{LLR}(\vx_{i}|\vy_{ir}) = 4h_i \sqrt{E} \vy_{ir}.
\end{equation}
As shown in \cite{Hagenauer1995}, the LLR of the network coded vector $\vx_1 \circ \vx_2$ based on $\vy_{1r}$ and $\vy_{2r}$ is given by
\begin{equation}\label{eq:LLR_xR}
\text{LLR}(\vx_1 \circ \vx_2|\vy_{1r},\vy_{2r})  =  \left\{\Sign(y_{1r,m} y_{2r,m}) \times
   \min\{ |4h_1 \sqrt{E} y_{1r,m}| , |4h_2 \sqrt{E} y_{2r,m}| \}\right \}_{m=1}^n.
\end{equation}
Methods of broadcasting the LLR vector of \eqref{eq:LLR_xR} to both end nodes have been studied widely. This is also known as soft network coding (SoftNC) \cite{Zhang2008,Yang2007}. They are quite similar to analog network coding (ANC), and thus also suffer from large peak-to-average power ratio (PAPR), noise propagation and other drawbacks associated with analog communication. One solution is to quantize each soft value and then broadcast the quantized information using high-level modulation. However, besides quantization errors, high-level modulation incurs larger transmission power in order to meet the BER requirement when using PSK, or larger bandwidth when FSK is employed.
%Provided that the decoding error occurs at $R$,
Our solution is to quantize each soft output of the relay decoder to one bit, exactly as in the case of TW-SDF. We then form $\vx_{r}$ as the network-coded combination of the two vectors of quantized bits. Thus, $\vx_{r}^{\text{TW-1bSF}}=\hat{\vx}_{1r}\circ\hat{\vx}_{2r}$. Unlike TW-SDF, $\vx_{r}^{\text{TW-1bSF}}$ is forwarded even when $\hat{\vx}_{1r}\neq \vx_1$ or $\hat{\vx}_{2r}\neq \vx_2$.

%Forwarding the vector of one-bit quantization of the soft outputs is equivalent to conveying their signs. To simplify the transmission of the reliability components,
Furthermore, $R$ generates a scalar value to measure the reliability of the network-coded vector, and then broadcasts it to the receivers (similar to the equivalent SNR approach in \cite{Amat2010}), or the receivers calculate the reliability value by themselves. Each receiver uses $\vy_{ri}$ and the associated reliability value as \emph{a priori} information to enhance the decoding performance of the signal vector received through the direct link. Thus the name \emph{two-way one-bit soft forwarding} (TW-1bSF) for the protocol that uses this relaying approach.

We note here that we can implement the TW-1bSF protocol in several ways.
%For example, the 1bSF protocol can be implemented using the source-controlled decoding method of \cite{Hagenauer1995} as following:
%\begin{equation}
%\hat{\vx}_{1}=\arg\max_{\vx\in\mathcal{X}^{c(n)}}\left(\langle 4h\sqrt{E}\vy_{12},\vx\rangle+\langle \mathcal{L}_u\hat{\vu},\vu_{\vx}\rangle \right)
%\end{equation}
%where where $\langle \cdot,\cdot \rangle$ is the inner product; $\parallel\text{\ensuremath{\cdot}}\parallel$ is the Euclidean norm; $\mathcal{X}^{c(n)}$ is the set of all valid codewords; $\vu_{\vx}$ is the information message associated with the codeword $\vx$; $\hat{\vu}$ is the sign function of the a posterior information of the message bits forwarded by $R$; and $\mathcal{L}_u$ is the reliability of $\hat{\vu}$ forwarded by $R$ or calculated by $S2$, by which the transmission of the reliability value is even saved for $R$.
However, in this paper, we use the simple implementation method described above, i.e., $\vx_{r}^{\text{TW-1bSF}}=\hat{\vx}_{1r}\circ\hat{\vx}_{2r}$ is always forwarded during the third time slot, unlike TW-SDF. An indicator bit is embedded in the overhead of the forwarded frame as, for example, in the COPE protocol \cite{Katti2006, Tran2010}. Then the receivers can calculate the reliability value based on the indicator bit and the channel information.

\subsection{Design of reliability value}

If $\vx_1$ can be successfully decoded from either of $\vy_{12}$ and $\tilde{\vy}_{r2}$, then the transmission succeeds. So, in the following we consider the detection of $\vx_1$ at $S2$ based on both $\vy_{12}$ and $\tilde{\vy}_{r2}$, assuming that decoding is not successful from either of them. The aim of TW-1bSF is to make $\vy_{r2}$ useful even if an error has occurred in $\hat{\vx}_{1r}$ or $\hat{\vx}_{2r}$. Let $\vl_{r2}$ be the LLR vector of the target message $\vx_1$ in TW-1bSF obtained by $\tilde{\vy}_{r2}$, i.e., $\vl_{r2} = \text{LLR}(\vx_1|\tilde{\vy}_{r2})$. By regarding $\vl_{r2}$ as the LLR of the $S1-R-S2$ path and using the decoding method of \cite{Hagenauer1995}, $S2$ can decode using the received signal vector $\vy_{12}$ as follows
\begin{align}\label{eq:optimal-decoder}
 \hat{\vx}_{1}
 = &  \argmax_{\vx\in\mathcal{X}^{c(n)}}\left(\langle4h\sqrt{E}\vy_{12},\vx\rangle+\langle \vl_{r2},\vx\rangle \right)   \cr
 = & \argmax_{\vx\in\mathcal{X}^{c(n)}}\left(4h\sqrt{E}\vy_{12}^T \vx+ \vl_{r2}^T \vx\right)   \cr
 = & \argmin_{\vx\in\mathcal{X}^{c(n)}}\left(- 4h\sqrt{E}\vy_{12}^T \vx - \vl_{r2}^T \vx\right)   \cr
 = & \argmin_{\vx\in\mathcal{X}^{c(n)}}(4h\sqrt{E}\vy_{12} + \vl_{r2}-\vx)^T (4h\sqrt{E}\vy_{12} + \vl_{r2}-\vx) \cr
 = & \argmin_{\vx\in\mathcal{X}^{c(n)}}\Vert 4h\sqrt{E}\vy_{12} + \vl_{r2}-\vx\Vert,
\end{align}
where $\langle \cdot,\cdot \rangle $ represents the inner product, $\Vert \cdot \Vert$ represents the $l_2$ norm, and $(\cdot)^T$ represents the matrix transposition.

Next, to design $\vl_{r2}$ we distinguish three different cases, depending on whether a decoding error occurs for $\hat{\vx}_{1r}$ and $\hat{\vx}_{2r}$ at $R$, and whether the forwarded network-coded packet $\vx_r$ is correctly decoded at the receiver.

Let $\mathfrak{D}(\vy)$ be a hard-output decoder. Then $\mathfrak{D}(\vy)\in \mathcal{X}^{n} $ for $\forall \vy$. We define the following erroneous decoding events
\begin{equation}\label{eq:err.events}
\begin{aligned}
\e\triangleq \{\hat{\vx}_1\neq \vx_1\}, \quad \e_{12}\triangleq \{\mathfrak{D}(\vy_{12})\neq \vx_1\},
\quad \e_{1r}\triangleq \{\hat{\vx}_{1r}\neq \vx_1\}, \\
\quad
\e_{2r}\triangleq \{\hat{\vx}_{2r}\neq \vx_2\},
\quad \e_{r2}\triangleq \{\mathfrak{D}(\tilde{\vy}_{r2})\neq \vx_r\circ \vx_2\}.
\end{aligned}
\end{equation}
In addition, let $A^c$ denote the complement of an error event $A$, i.e. the event of correct decoding, for example, $\e^c=\{\hat{\vx}_1 = \vx_1\}$.

In the following, we introduce the three cases and the corresponding design of $\vl_{r2}$.
\begin{itemize}
\item \textbf{Case I}: This is the case where $\hat{\vx}_{1r}$ and $\hat{\vx}_{2r}$ are both correctly decoded, i.e., $\hat{\vx}_{1r}=\vx_{1}$ and $\hat{\vx}_{2r}=\vx_{2}$, but the forwarded signal $\vx_r$ is not correctly decoded at $S2$, i.e., $\mathfrak{D}(\vy_{r2})\neq \vx_r$. It is denoted by $\Theta_\co\triangleq \e_{12}\cap \e^c_{1r}\cap \e^c_{2r}\cap \e_{r2}$. In this case, from \eqref{eq:def-yr}, $\tilde{\vy}_{r2} = h_{2}\sqrt{E_r}\vx_{1}+\tilde{\vn}_{r2}$, so we simply have
\begin{equation}\label{eq:Lr2CaseI}
\vl_{r2}=\left\{\text{LLR}(x_{1,m}|\tilde{y}_{r2,m})\right\}_{m=1}^n =4h_2\sqrt{E_r}\tilde{\vy}_{r2}.
\end{equation}
Comparing \eqref{eq:Lr2CaseI} with \eqref{eq:decTWSDF} shows that in this case TW-1bSF and TW-SDF have the same decoding results.
\item \textbf{Case II}: This is the case where at least one of $\hat{\vx}_{1r}$ and $\hat{\vx}_{2r}$ is incorrectly decoded, i.e., $\e_{1r}\cup \e_{2r}$, but $\vx_r$ is correctly decoded at $S2$, and is denoted by $\Theta_\cc \triangleq \e_{12}\cap(\e_{1r}\cup \e_{2r})\cap\e^c_{r2}$. In this case, we let
\begin{equation}\label{eq:Lyr-caseII}
\vl_{r2}=(\mathfrak{D}(\vy_{r2})\circ \vx_2)\mathcal{L} = (\vx_r\circ \vx_2)\mathcal{L},
\end{equation}
where $\mathcal{L}$ is a scalar value measuring the reliability of the forwarded vector $\vx_r$. This case can be further decomposed into three disjoint subcases as $\Theta_\cc  = \Theta_{\cc a}\cup \Theta_{\cc b} \cup \Theta_{\cc c}$, where $\Theta_{\cc a}  \triangleq \e_{12}\cap \e_{1r}\cap \e^c_{2r}\cap\e^c_{r2}$, $\Theta_{\cc b}  \triangleq \e_{12}\cap \e^c_{1r}\cap \e_{2r}\cap\e^c_{r2}$ and $\Theta_{\cc c}  \triangleq \e_{12}\cap \e_{1r}\cap \e_{2r}\cap\e^c_{r2}$.
\item \textbf{Case III}: This is the case where at least one of $\hat{\vx}_{1r}$ and $\hat{\vx}_{2r}$ is incorrectly decoded, and $\vx_r$ is not correctly decoded at $S2$ either. It is denoted by $\Theta_{\ccc}\triangleq \e_{12}\cap(\e_{1r}\cup \e_{2r})\cap \e_{r2}$. In this case, we let
\begin{equation}
\vl_{r2}=\mathfrak{D}(\tilde{\vy}_{r2})\mathcal{L}^-,
\end{equation}
where $\mathcal{L}^-$ is the scalar reliability value for Case III.
\end{itemize}
%
%%According to the principle of LLR, $\mathcal{L}$ and $\mathcal{L}^-$ should be inversely proportional to the codeword error rate.
%%Using the approximation of the codeword error rate by the channel bit error rate, and then referring to the relationship between bit error rate and LLR as well as the LLR operation rule \cite{Hagenauer1995},
In summary,
\begin{equation}
\vl_{r2}=
\begin{cases}
4h_2\sqrt{E_r}\tilde{\vy}_{r2}, & \text{for Case I} \\
(\vx_r\circ \vx_2)\mathcal{L}, & \text{for Case II} \\
\mathfrak{D}(\tilde{\vy}_{r2})\mathcal{L}^-, & \text{for Case III}
\end{cases}.
\end{equation}

We then design $\mathcal{L}$ and $\mathcal{L}^-$. According to \eqref{eq:LLR_xR}, the scalar reliability value $\mathcal{L}$ of Case II should be a function of the reliability values $\{ 4\sqrt{E} \min\{ |h_1 y_{1r,m}| , |h_2 y_{2r,m}| \}\}_{m=1,2,\dots,n}$. Note that
\begin{equation}\label{eq:min_LLR}
4\sqrt{E} \min\{ |h_1 y_{1r,m}| , | h_2 y_{2r,m} | \}
=  \min \left\{  \log \tfrac{\exp(-\frac{(y_{1r,m}-h_1\sqrt{E})^2}{2\sigma_r^2}) }
{\exp(-\frac{(y_{1r,m}+h_1\sqrt{E})^2}{2\sigma_r^2}) } ,
\log \tfrac{\exp(-\frac{(y_{2r,m}-h_2\sqrt{E})^2}{2\sigma_r^2}) }
{\exp(-\frac{(y_{2r,m}+h_2\sqrt{E})^2}{2\sigma_r^2}) } \right\}.
\end{equation}
We have a few options on choosing $\mathcal{L}$, such as taking the arithmetic average of the reliability values. Alternatively, simply exploiting the relationship between the LLR and the raw bit error rate,
\begin{equation}\label{eq:sum_LLR}
\frac{1-p_{1r}}{p_{1r}}=
\tfrac{\int_0^{\infty}\tfrac{1}{\sqrt{2\pi \sigma_r^2}}\exp\left(-\frac{(y_{1r,m}-h_1\sqrt{E})^2}{2\sigma_r^2}\right) \text{d} y_{1r,m} }
{\int_0^{\infty}\tfrac{1}{\sqrt{2\pi \sigma_r^2}} \exp\left(-\frac{(y_{1r,m}+h_1\sqrt{E})^2}{2\sigma_r^2}\right) \text{d} y_{1r,m} },
\end{equation}
and we can set the reliability value for Case II as
\begin{equation}\label{eq:rel-old}
\mathcal{L}=\min\left(\log\frac{1-p_{1r}}{p_{1r}},\log\frac{1-p_{2r}}{p_{2r}}\right),
\end{equation}
where $p_{1r}=Q(\sqrt{2h_{1}^{2}E})$ is the channel bit error rate
of the $S1$-$R$ link, $p_{2r}=Q(\sqrt{2h_{2}^{2}E})$ is
the channel bit error rate of the $S2$-$R$ link, and $Q(\cdot)$ is the Gaussian Q function.
In \eqref{eq:min_LLR}, the likelihood is measured using the
probability density function,  whereas in \eqref{eq:sum_LLR} the likelihood is given by
the cumulative density function, which is unrelated to the individual received symbols $y_{1r,m}$ (or $y_{2r,m}$).

Using the same principle as in \eqref{eq:rel-old}, the reliability value for Case III can be set to
\begin{equation}\label{eq:rel-case3}
\mathcal{L}^-  =  \log\frac{1-(1-(1-p_{1r})(1-p_{2r}))}{1-(1-p_{1r})(1-p_{2r})}
 =  \log\frac{(1-p_{1r})(1-p_{2r})}{1-(1-p_{1r})(1-p_{2r})},
\end{equation}
where $1-(1-p_{1r})(1-p_{2r})$ is the probability that at least one of the $S1$-$R$ and the $S2$-$R$ links is in error. Here we choose $\mathcal{L}^-$ intuitively. In Section V we will discuss the impact of the choice of $\mathcal{L}$ and $\mathcal{L}^-$ on the error performance.

\begin{rem}
As in COPE \cite{Katti2006}, the end nodes use an indicator bit embedded in the overhead and the CRC associated with $\vx_{r}^{\text{TW-1bSF}}$ to identify the cases that the forwarded message belongs to. Specifically, if the indicator bit is 1, it signals Case I, otherwise Case II or Case III. When the indicator bit is 0, if the received signal vector $\vy_{r2}$ passes the CRC check, then the destination proceeds according to Case II, otherwise Case III. Here we assume that each end node knows the channel coefficients $h_1$ and $h_2$. Hence, it can obtain $\mathcal{L}$ and $\mathcal{L}^-$ by calculating $p_{1r}$ and $p_{2r}$ from the link coefficients and then substituting them into \eqref{eq:rel-old} and \eqref{eq:rel-case3}.
\end{rem}

\subsection{Selection of the Channel Code and the Decoder}
Many channel codes are applicable to the proposed TW-1bSF protocol.
%For example, if the convolutional code is used, $R$ could use the BCJR algorithm to calculate the a posterior probabilities of messages bits and then re-encode the decoded message to a codeword, or use the hard-output Viterbi decoding method to directly achieve a decoded codeword. At the receiver end, the received signal $\tilde{\vy}_{r2}$ could be decoded first to obtain the soft information of message bits/code bits, and then the soft values are fed into the source controlled soft-input Viterbi decoder as a priori source information to help decode the received signal $\vy_{12}$ from the direct link \cite{Hagenauer1995}. Or the received signal can be directly used at the minimum distance decoder as listed in \eqref{eq:MDD-ibSF} by using Viterbi algorithm.
However, for WSNs with limited energy and computation capability, simple block codes like Hamming or BCH codes are preferable. In particular, BCH
codes with a syndrome decoder that uses the Berlekamp-Massey (BM) and Chien's search (CS) algorithm are 15\%
more energy efficient than the best performing convolutional codes
\cite{Sankarasubramaniam2003}. Therefore, in the sequel, we study the TW-1bSF protocol when Hamming or BCH codes are applied, and hard decoders are employed at the receiving ends. Consequently, the decoder \eqref{eq:optimal-decoder} can be rewritten as
\begin{equation}
\hat{\vx}_{1}^{\text{TW-1bSF}}=\mathfrak{D}(\vl_d),
\end{equation}
where $\vl_d \triangleq 4h\sqrt{E}\vy_{12} + \vl_{r2}$.

For most practical block codes, the minimum distance is an odd number, conventionally denoted by $2t+1$, where $t$ is known as the correcting capacity. Therefore, in the following analysis we consider block codes with minimum distance $2t+1$.

\section{\label{sec:Performance-Analysis}Performance Analysis of TW-SDF and TW-1bSF}
This section analyzes the performance of the TW-SDF and the TW-1bSF protocols in terms of block error rate (BLER). It is assumed that the messages are encoded using block codes, whereas the receiving sources employ hard decoders, as described in Section~III.

%\begin{defn}
Without loss of generality, we consider the detection of $\vx_1$ at $S2$. We derive the BLER of the codewords transmitted from $S1$ to $S2$, defined as
\begin{equation}
P^{(\vx_1)}\triangleq \Pr(\hat{\vx}_{1}\neq\vx_{1}).
\end{equation}
%\end{defn}
%\begin{rem}
%Note that one may want to study the BLER taking into account
%both $\vx_{1}$ and $\vx_{2}$, and define it like $1-\Pr(\tilde{\vx}_{1}=\vx_{1},\tilde{\vx}_{2}=\vx_{2})$.
%In the following analysis, it will show the derivation of
%$P^{(\vx_1)}$ is more fundamental, and can be straightly extended
%to the derivation of other kinds of BLER definition.
%\end{rem}

%\noindent \textbf{Preliminary}:
\subsection{Background}
Recall that in the two-way relaying network, there are four point-to-point links associated with the transmission of $\vx_{1}$, i.e., $S1\rightarrow S2$, $S1\rightarrow R$, $S2\rightarrow R$ and $R\rightarrow S2$, denoted as $\{12,1r,2r,r2\}$, respectively. We first derive the BLER of these links based on their raw bit error rates  $p_{12}=Q(\sqrt{2h^2E})$, $p_{1r}=Q(\sqrt{2h_1^2E})$, $p_{2r}=Q(\sqrt{2h_2^2E})$ and $p_{r2}=Q(\sqrt{2h_2^2E_r})$. No matter which specific hard decoding method is applied, we can upper and lower bound the BLER of a link by using the performance of the \emph{bounded distance decoder} (BDD). Assuming that a decoding error occurs if the received word falls \emph{outside} the decoding sphere of the \emph{transmitted} codeword \cite{Moon2005}, we obtain the upper bound
\begin{equation}
\stackrel{_\frown}{P}_i=1-\sum_{k=0}^{t}\binom{n}{k}p_i^{k}(1-p_i)^{(n-k)},
\label{eq:BLER-upper}
\end{equation}
where $i\in \{12,1r,2r,r2\}$.
On the other hand, assuming that a decoding error occurs only if the received words fall \emph{inside} the decoding spheres of the codewords that have not been transmitted, we obtain the lower bound
\begin{equation}
\stackrel{_\smile}{P}'_i  = \sum_{k=2t+1}^{n}\beta'_{k}(p_i)
 =  \sum_{k=2t+1}^{n}A_{k}\sum_{m=0}^{t}\sum_{j=0}^{\min(m,n-k)}\binom{k}{m-j}\binom{n-k}{j}
\times p_i^{k-m+2j}(1-p_i)^{n-k+m-2j},
\label{eq:BLER-lower}
\end{equation}
where $A_k$ is the number of the codewords with Hamming weight $k$
%, which is conventionally expressed in the weight enumerate function (WEF) $\sum_{k=2t+1}^{n}A_{k}D^{k}$
; $\beta'_k(p_i)$ lower bounds the probability that the codeword to which the received word is decoded has Hamming distance $k$ from the transmitted codeword.
Note that in \eqref{eq:BLER-lower} we have assumed, without loss of generality, that the all-zero codeword has been transmitted.
However, except for perfect codes, the lower bound \eqref{eq:BLER-lower} is too loose to accurately predict the BLER of any practical hard decoder. This is because it does not take into account undecodable sequences that are outside the Hamming spheres of any codewords. Among the dominant part of the undecodable received words, i.e., the $\binom{n}{t+1}$ received words of weight $t+1$, the lower bound \eqref{eq:BLER-lower} only accounts for $A_{2t+1}\binom{2t+1}{t}$ words that are inside the Hamming spheres of the codewords of weight $2t+1$.

\subsection{A tighter lower bound}
Here, we propose a lower bound that is tighter than \eqref{eq:BLER-lower}. Suppose, again, that the all-zero codeword is transmitted. Since the correcting capability is $t$, the received words of Hamming weight $t+1$ are not guaranteed to be correctly decoded, and which codewords they will be decoded to depends on the decoding method. We denote the number of the received words of weight $t+1$ that are decoded to wrong codewords of weight $k$ by $W_k$. We call $W_k$, $2t+1\leq k \leq n$ the \emph{sphere partitioning function} (SPF). The SPF partitions the surface of the radius-$(t+1)$ sphere centered at the transmitted zero codeword. Note that $W_0 \neq 0$ for almost all codes and hard decoders. If $W_k$ for a given hard decoder $\mathfrak{D}(\cdot)$ is known,
%In practice, the weight function can be achieved offline by decoding the $\binom{n}{t+1}$ words of weight $t+1$ or simply by %simulation.
the following lemma provides a tighter lower bound on the BLER of the hard decoder.
%\begin{lemma}\label{lemma:lowerbd.of.UBA}
%Let $V=\tfrac{\binom{n}{t+1}}{A_{2t+1}\binom{2t+1}{t}}$, and $\varrho_i^{(v)}=(V-1)A_{2t+1}\binom{2t+1}{t}p_i^{t+1}(1-p_i)^{n-(t+1)}$.Then
%\begin{eqnarray}
%PB_i^{(u)}& \geq & PB_i^{(l)} + \varrho_i^{(v)} \nonumber \\
%&=& \sum_{k=2t+1}^{n}P_k(p_i) + I_{2t+1}(k) \varrho_i^{(v)}
%\label{eq:lowerbd.of.UBA}
%\end{eqnarray}
%where $I_{2t+1}(d) = 1$ if $d=2t+1$, $I_{2t+1}(d) =0$ otherwise. And
%\begin{equation}
%\lim_{p_i\rightarrow 0}\frac{PB_i^{(u)}}{PB_i^{(l)}+\varrho_i^{(v)}} =1
%\end{equation}
%In other words, $PB_i^{(u)}$ is lower bounded by, and also asymptotically equal to $PB_i^{(l)}+\varrho_i^{(v)}$.
%\end{lemma}
%Relatively, compared to the upper bound, the lower bound is loose, but will be very useful in the sequel.
\begin{lemma}[\emph{Intrinsic Lower Bound of Hard Decoder}]\label{lemma:lowerbd.of.UBA}
Given a hard decoder $\mathfrak{D}(\cdot)$ of SPF $W_k$ and a link of raw BER $p_i$, the BLER of the link when $\mathfrak{D}(\cdot)$ is used is lower bounded by \eqref{eq:de-BLER-lower}.
\begin{figure*}
\begin{equation}
\stackrel{_\smile}{P}_i =  \underset{\beta_{2t+1}(p_i)}{\underbrace{\beta'_{2t+1}(p_i) + \left(W_{2t+1} - A_{2t+1}\binom{2t+1}{t}\right)p_i^{t+1}(1-p_i)^{n-t-1}}}
+ \sum_{k=2t+2 }^n \underset{\beta_k(p_i)}{\underbrace{\beta'_k(p_i) +
W_k p_i^{t+1}(1-p_i)^{n-t-1}}}
\label{eq:de-BLER-lower}
\end{equation}
\end{figure*}
\end{lemma}
\begin{IEEEproof}
Given the transmission of the zero codeword, consider all the received words of weight $t+1$. If the hard decoder's SPF $W_k$ is known, it is clear that except for the $A_{2t+1}\binom{2t+1}{t}$ words included in \eqref{eq:BLER-lower}, additional $W_{2t+1} - A_{2t+1}\binom{2t+1}{t}$ words are decoded to wrong codewords of weight $2t+1$, and $\sum_{k=2t+2}^n W_k$ words are decoded to wrong codewords of weight $k > 2t+1$. This concludes the proof of the lemma.
\end{IEEEproof}
Note that with $\beta_k(p_i)$ defined as in \eqref{eq:de-BLER-lower}, $\beta_k(p_i)>\beta'_k(p_i)$, $k\in\{2t+1,\cdots,n\}$, so the lower bound \eqref{eq:de-BLER-lower} is tighter than \eqref{eq:BLER-lower}.

\subsection{Preliminary results}
Recall that all receivers utilize the same energy-efficient hard decoder with SPF $W_k$.
%{\color{blue} Hereafter, let $d(\vy)=\mid \{\tau\in\{1,2,\cdots,n\}:y_{\tau} <0\}\mid $ be the number of negative signals within a signal vector $\vy$; $d(\vy_1,\vy_2) = \mid \{\tau:y_{1,\tau}y_{2,\tau} <0\}\mid $ denotes the number of signals with different signs between two signal vectors $\vy_1$ and $\vy_2$.}
%
%\subsection{Performance of 1bSF protocol}
%
Due to the equiprobability of the codewords and the linearity of block codes, the BLER of $\vx_1$ at $S2$ equals the conditional BLER given the transmission of $\vx_1=\vx_2=\{+1\}_{m=1}^n=\mathbf{1}$, which are the BPSK modulated signal vectors that correspond to the all-zero codewords. % $\vc_1 =\vc_2 = \mathbf{0}$.
The following analysis is carried out under the assumption that $\vx_1=\vx_2=\mathbf{1}$, which is omitted from now on for simplicity.

By decomposing TW-SDF transmission into disjoint cases, its BLER can be expressed as
\begin{equation}\label{eq:PB-SDF}
P_{\text{TW-SDF}}^{(\vx_1)}  = P(\e,\Theta_{\co})+P(\e,\Theta_{\cc}\cup\Theta_{\ccc})
 = P(\e,\Theta_{\co})+P(\Theta_{\cc}\cup\Theta_{\ccc}).
\end{equation}
Similarly, for TW-1bSF,
\begin{equation}\label{eq:PB-1bSF}
P^{(\vx_1)}_{\text{TW-1bSF}} =  P(\e,\Theta_\co\cup \Theta_\cc \cup \Theta_\ccc)
= P(\e,\Theta_\co)+P(\e,\Theta_\cc)+P(\e,\Theta_\ccc).
\end{equation}

We define $d(\vy)\triangleq \mid \{\tau\in\{1,2,\cdots,n\}:y_{\tau} <0\}\mid $ as the number of negative elements in a vector (or scalar) $\vy$. Because $\vx_1=\vx_2=\mathbf{1}$, the following three properties will be used in the following BLER analysis.
\begin{enumerate}
\item A raw bit is \emph{in error} if the corresponding signal $y_{\tau}<0$ for $\tau\in\{1,\cdots,n\}$, i.e., $d(y_{\tau})=1$;
\item If the received signal vector $\vy$ is decoded erroneously, then the decoded codeword has weight at least $2t+1$, i.e., $d(\D(\vy))\geq 2t+1$;
\item When a Bounded Distance Decoder (BDD) is used at receiver $i$, $\vy_i$ is correctly decoded iff $d(\vy_i)\leq t$.
\end{enumerate}
To distinguish from the events in \eqref{eq:err.events}, if a BDD is used, we denote the decoding error event at the receiver of link $i$ as $\grave{\e}_i$. Then the probability of this error event is $P(\grave{\e}_i) = \stackrel{_\frown}{P}_i $ as given by \eqref{eq:BLER-upper}. Also denote
\begin{equation}
\grave{\Theta}_{j}\triangleq \{\Theta_{j} : \text{$\D$ is a BDD for link $i\in\{12,1r,2r,r2\}$} \}
\end{equation}
with $j\in\{\co,\cc,\ccc\}$. Specifically, $\grave{\Theta}_\co =  \grave{\e}_{12}\cap \grave{\e}^c_{1r}\cap \grave{\e}^c_{2r}\cap \grave{\e}_{r2}$, $ \grave{\Theta}_{\cc} = \grave{\e}_{12}\cap(\grave{\e}_{1r}\cup \grave{\e}_{2r})\cap\grave{\e}^c_{r2}$ and $ \grave{\Theta}_{\ccc} = \grave{\e}_{12}\cap(\grave{\e}_{1r}\cup \grave{\e}_{2r})\cap\grave{\e}_{r2}$.
%To simplify the analysis, we study the worst case in which the sphere decoder, and the corresponding upper bound given by \eqref{eq:BLER-upper}, is applied on decoding the received signals at $S2$. \\
%\noindent \textbf{Bounds}: (1) UBA is used for decoding $\vy_{12}$, in which $\vy_{12}$ is decoded in error if $\vy_{12}$ has more than $t$ erroneous channel bits, i.e, $P(\e_{12})\leq P(\{d(\vy_{12})\geq (t+1)\})$;
%%(2) LBA is used for decoding $\vy_{1r}$, $\vy_2^{(2)}$ and $\tilde{\vy}_{r2}$, for instance $\tilde{\vy}_{r2}$ is decoded as a codeword $\vx^{(k)}$ with weight $k$ if $\tilde{\vy}_{r2}$ falls inside the decoding sphere of the codeword, so, $\e_{r2}=\{\D_\hsd(\tilde{\vy}_{r2})=\vx^{(k)},k\neq 0\}\approx \{d(\tilde{\vy}_{r2},\vx^{(k)})\leq t,k\neq 0\}$; and
%(2) UBA is used for decoding $\vl_d$, in which $\vl_d$ is correctly decoded if $\vl_d$ falls inside the decoding sphere of the all-one symbol word, i.e., $P(\e)\leq P(\{d(\vl_d)\geq (t+1)\})$, and hence $P({\e}^c)\geq P(\{d(\vl_d)\leq t\})$;

%\subsubsection{BLER of Case I}
We start with the analysis of $P(\e,\Theta_{\co}) $, since it is the common term of $P_{\text{TW-SDF}}^{(\vx_1)}$ and $P_{\text{TW-1bSF}}^{(\vx_1)}$. The event $\Theta_{\co}$ is decoder-dependent, and its probability is usually hard to determine, if not impossible. Alternatively, we examine its worst-decoder counterpart, $\grave{\Theta}_{\co}$, and derive a upper bound of $P(\e,\grave{\Theta}_{\co}) $.
%, which would be late used to derive the performance bounds for $P_{\text{SDF}}^{(\vx_1)}$ and $P_{\text{1bSF}}^{(\vx_1)}$.
From the definition of $\grave{\Theta}_{\co}$, and the fact that the BDD is the worst hard decoder in terms of error performance, we have
%\begin{align}\label{eq:bler-caseI}
%& P(\e,\grave{\Theta}_{\co}) \cr
%&\leq P(\grave{\e},\grave{\Theta}_{\co}) \cr &=P(\grave{\Theta}_{\co})-P(\grave{\e}^c,\grave{\Theta}_{\co}) \cr
%& = \stackrel{_\frown}{P}_{12}(1-\stackrel{_\frown}{P}_{1r})
%(1-\stackrel{_\frown}{P}_{2r})\stackrel{_\frown}{P}_{r2}-P(\grave{\e}^c,\grave{\Theta}_{\co}).
%\end{align}
\begin{equation}\label{eq:bler-caseI}
 P(\e,\grave{\Theta}_{\co})
\leq P(\grave{\e},\grave{\Theta}_{\co})  =P(\grave{\Theta}_{\co})-P(\grave{\e}^c,\grave{\Theta}_{\co})
 = \stackrel{_\frown}{P}_{12}(1-\stackrel{_\frown}{P}_{1r})
(1-\stackrel{_\frown}{P}_{2r})\stackrel{_\frown}{P}_{r2}-P(\grave{\e}^c,\grave{\Theta}_{\co}).
\end{equation}
The term $P(\grave{\e}^c,\grave{\Theta}_{\co})$ in \eqref{eq:bler-caseI} is given by the following theorem.
\begin{thm}\label{thm:ineq-noe-CaseI}
$P(\grave{\e}^c,\grave{\Theta}_{\co})$
can be calculated as \eqref{eq:lb-noe-CaseI},
\begin{eqnarray}\label{eq:lb-noe-CaseI}
P(\grave{\e}^c,\grave{\Theta}_{\co}) & \triangleq &    (1-\stackrel{_\frown}{P}_{1r})(1-\stackrel{_\frown}{P}_{2r})\sum_{k=t+1}^{n-1}\binom{n}{k} \sum_{m=\max(0,t+1-(n-k))}^{t}\sum_{i=t+1}^{m+n-k}\binom{k}{m}\binom{n-k}{i-m}
 p_{\co}^{n-k-(i-m)}q_{\co}^{m}\nonumber \\
 &  & \makebox[3cm]{} \times\sum_{g=0}^{t-m}\sum_{j=0}^{g}\binom{k-m}{j}
\dot{p}_{\co}^{j}\dot{q}_{\co}^{k-m-j}\binom{i-m}{g-j}\ddot{p}_{\co}^{g-j}\ddot{q}_{\co}^{i-m-(g-j)}.
\end{eqnarray}
where
%\begin{subequations}
%\begin{eqnarray}
% p_{\co} & = & \left(1-p_{r2}\right)\left(1-p_{12}\right), \\
%\quad q_{\co} &= & p_{12}p_{r2}, \\
%%
%\dot{p}_{{\co}} &=& \text{mvnCDF}([0,0],\bmu_1,\bSigma_1), \label{eq:d_p_co}\\
%\dot{q}_{\co} &=& p_{r2}(1-p_{12})-\dot{p}_{\co}, \\
%%
%\ddot{p}_{\co} &=& \text{mvnCDF}([0,0],\bmu_2,\bSigma_2), \label{eq:dd_p_co} \\
%\ddot{q}_{\co}& =&  p_{12}(1-p_{r2})-\ddot{p}_{\co}.
%\end{eqnarray}
%\end{subequations}
\begin{equation}\label{eq:p_co_eqs}
\begin{cases}
 p_{\co}  =  \left(1-p_{r2}\right)\left(1-p_{12}\right), \quad
 q_{\co} =  p_{12}p_{r2}, \quad
\dot{p}_{{\co}} = \text{mvnCDF}([0,0],\bmu_1,\bSigma_1), \\
 \dot{q}_{\co} = p_{r2}(1-p_{12})-\dot{p}_{\co}, \quad
\ddot{p}_{\co} = \text{mvnCDF}([0,0],\bmu_2,\bSigma_2),  \quad
 \ddot{q}_{\co} =  p_{12}(1-p_{r2})-\ddot{p}_{\co}.
\end{cases}
\end{equation}
In \eqref{eq:p_co_eqs}, $\text{mvnCDF}(\cdot,\cdot,\cdot)$ represents the \emph{multivariate normal cumulative distribution function} using Matlab notation, with mean vectors $\bmu_1=[a^2 + b^2, -b]$ and
$\bmu_2 = [a^2 + b^2, -a]$, and covariance matrices
$\bSigma_1 = \tfrac{1}{2}[a^2+b^2, -b; -b,1]$ and $\bSigma_2 = \tfrac{1}{2}[a^2+b^2, -a; -a,1]$, in which $a=h_2\sqrt{E_r}$ and $b=h\sqrt{E}$.
\end{thm}
\begin{IEEEproof}
Recall that $\grave{\Theta}_\co =  \grave{\e}_{12}\cap \grave{\e}^c_{1r}\cap \grave{\e}^c_{2r}\cap \grave{\e}_{r2}$. Because of $\grave{\e}^c_{1r}$ and $\grave{\e}^c_{2r}$, $\tilde{\vy}_{r2} = h_{2}\sqrt{E_r}\vx_{1}+\tilde{\vn}_{r2}$ according to \eqref{eq:def-yr}, so $\grave{\e}_{r2}$ means that $\tilde{\vy}_{r2}$ has $d(\tilde{\vy}_{r2})= k \geq t+1$ bits in error. Similarly, $\grave{\e}_{12}$ means that $\vy_{12}$ has $i\geq t+1$ bits in error.
Recall that in Case I, $\vl_d = 4h\sqrt{E}\vy_{12}+4h_2\sqrt{E_r}\tilde{\vy}_{r2}$. Using the fact that $P(\grave{\e}^c, \grave{\Theta}_{\co})= P(d(\vl_d)\leq t, \grave{\Theta}_{\co})$, we have
\begin{equation}\label{eq:ineq-noe-CaseI}
P(\grave{\e}^c,\grave{\Theta}_{\co})  = (1-\stackrel{_\frown}{P}_{1r})(1-\stackrel{_\frown}{P}_{2r})\sum_{k=t+1}^{n-1}\binom{n}{k}
 P(d(\vl_d)\leq t,d(\tilde{\vy}_{r2})=k,d(\vy_{12})\geq t+1).
\end{equation}
Note that here $k=n$ is omitted because $P(d(\vl_d)\leq t,d(\tilde{\vy}_{r2})=n,d(\vy_{12})\geq t+1 ) =0$. A detailed expression for \eqref{eq:ineq-noe-CaseI} is derived in Appendix~\ref{app:proof:theorem:caseI}, which proves the theorem.
\end{IEEEproof}
%By using the theorem, we upper bound $P(\e,\grave{\Theta}_{\co})$ by
%\begin{equation}\label{eq:bler-caseI}
%{\color{blue}P(\e,\grave{\Theta}_{\co})\leq P(\grave{\e},\grave{\Theta}_{\co}) \triangleq \stackrel{_\frown}{P}_{12}(1-\stackrel{_\frown}{P}_{1r})
%(1-\stackrel{_\frown}{P}_{2r})\stackrel{_\frown}{P}_{r2}-P(\grave{\e}^c,\grave{\Theta}_{\co}) .}
%\end{equation}
As to the numerical evaluation of the lower bound \eqref{eq:lb-noe-CaseI}, it is accurate enough to calculate only its first few terms, because the first few terms correspond to the received words containing a small number of erroneous bits (referring to $k$ and $i$ in Fig.~\ref{fig:case1-block}) and thus dominate the error event. This truncated calculation results in the obtained value being slightly loose, but the difference is negligible, and the truncated results are still valid lower bounds. In particular, our experiments show that the calculation based on truncation gives accurate results when truncating at $k,i=10$ when the code length $n\leq 127$.

\begin{figure}%[h]
\center\includegraphics[scale=0.45]{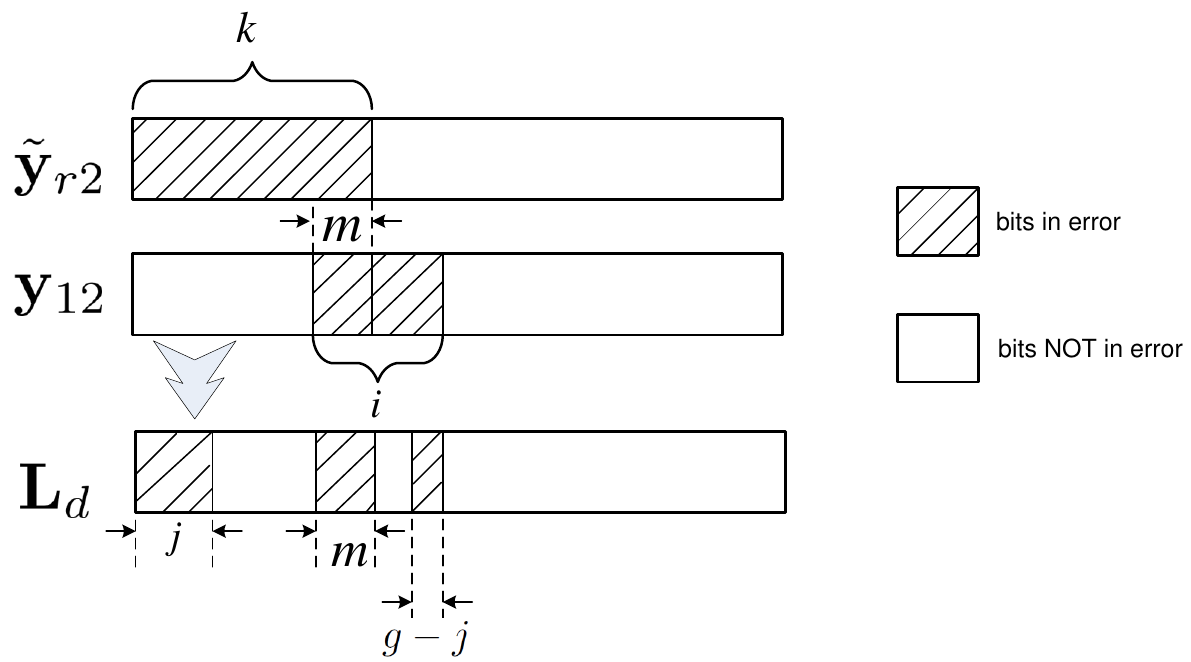}
\caption{\label{fig:case1-block} The alignment between $\tilde{\vy}_{r2}$ and $\vy_{12}$ for decoding $\vl_d=4h\sqrt{E}\vy_{12}+4h_2\sqrt{E_r}\tilde{\vy}_{r2}$ in Case I. The erroneous channel bits are represented by shadowed boxes, whereas the correct channel bits are represented by blank boxes.}
\end{figure}

\subsection{Performance of TW-SDF}
The \emph{one-way} SDF protocol has been well studied in terms of outage probability, symbol error rate and ergodic capacity for the scenario where no channel coding is applied at the relays, but few results are available for scenarios where channel coding is employed by the relaying networks. In \cite{Amat2010}, a pair of loose BLER bounds is derived for the cooperative DF protocol \cite{Wang2007} when convolutional codes and \emph{soft} decoders are used. In this section, we contribute new results for relaying networks that use channel coding by deriving a tight upper bound for the DNC-based TW-SDF protocol that uses block codes and \emph{hard} decoders.
%
% Recall
%$P_{\text{TW-SDF}}^{(\vx_1)}=P(\e,\Theta_{\co})+P(\Theta_{\cc}\cup\Theta_{\ccc})$, in the following theorem we derive its upper bound.
\begin{thm}
\label{thm:BLER_SDF_upb}
Suppose that all nodes apply a hard decoder whose performance is at least as good as BDD. Then
the BLER of the TW-SDF protocol can be upper bounded as
\begin{equation}\label{eq:bler-SDF}
P_{\text{TW-SDF}}^{(\vx_1)} \leq  \stackrel{_\frown}{P}_{\text{TW-SDF}}^{(\vx_1)} \triangleq
P(\grave{\Theta}_{\cc}\cup \grave{\Theta}_{\ccc}) + P(\grave{\e},\grave{\Theta}_{\co}),
\end{equation}
where  $P(\grave{\e},\grave{\Theta}_{\co})$ is given by \eqref{eq:bler-caseI}.
\end{thm}
\begin{IEEEproof}
See Appendix~\ref{app:proof:thm:BLER_SDF_upb}.
\end{IEEEproof}

Note that $\Theta_{\co} \nsubseteq \grave{\Theta}_{\co}$, so the proof of inequality \eqref{eq:bler-SDF} is not straightforward. Also note that
the upper bound $\stackrel{_\frown}{P}_{\text{TW-SDF}}^{(\vx_1)} $ is not related with the sphere partition function (SPF) $W_k$; thus it is decoder-independent.

From the definitions $ \grave{\Theta}_{\cc} = \grave{\e}_{12}\cap(\grave{\e}_{1r}\cup \grave{\e}_{2r})\cap\grave{\e}^c_{r2}$ and $ \grave{\Theta}_{\ccc} = \grave{\e}_{12}\cap(\grave{\e}_{1r}\cup \grave{\e}_{2r})\cap\grave{\e}_{r2}$, $P(\grave{\Theta}_{\cc} \cup \grave{\Theta}_{\ccc}) = P(\grave{\Theta}_{\cc}) + P(\grave{\Theta}_{\ccc})$ in \eqref{eq:bler-SDF} can be derived using $  P(\grave{\Theta}_{\cc}) = \stackrel{_\frown}{P}_{12}(\stackrel{_\frown}{P}_{1r}+\stackrel{_\frown}{P}_{2r} - \stackrel{_\frown}{P}_{1r}\stackrel{_\frown}{P}_{2r})(1-\stackrel{_\frown}{P}_{r2})$ and $ P(\grave{\Theta}_{\ccc}) = \stackrel{_\frown}{P}_{12}(\stackrel{_\frown}{P}_{1r}+\stackrel{_\frown}{P}_{2r} - \stackrel{_\frown}{P}_{1r}\stackrel{_\frown}{P}_{2r})\stackrel{_\frown}{P}_{r2}$. Simulation results in Section~V will verify the tightness of the upper bound of Theorem~\ref{thm:BLER_SDF_upb}.
%Note the extension of the following result to the one-way SDF protocol is straightforward.

\subsection{Performance of TW-1bSF}

 In general, deriving the exact BLER performance of the TW-1bSF protocol for a given hard decoder is a formidable task, except for some specific cases of special hard decoders like BDD. Alternatively, we resort to deriving upper bounds on the performance. We can simply universally upper bound $P^{(\vx_1)}_{\text{TW-1bSF}}=P(\e, \Theta_\co\cup \Theta_\cc\cup \Theta_\ccc)$ by  $\max_{\D}P(\e, \Theta_\co\cup \Theta_\cc\cup \Theta_\ccc)=P(\e, \Theta_\co\cup \Theta_\cc\cup \Theta_\ccc)|_{\D=\text{BDD}}$, in which all links and $\vl_d$ are decoded using BDD. However, this upper bound is very loose for most practical decoders. In the following, we derive a tighter upper bound.
\begin{lemma}[\emph{Upper Bound for the BLER of TW-1bSF}]
Given any hard decoder, we have
\begin{equation}\label{eq:e.event.bounds}
P^{(\vx_1)}_{\text{TW-1bSF}}  =  P(\e, \Theta_\co\cup \Theta_\cc\cup \Theta_\ccc )
  \leq   P(\e,\grave{\Theta}_{\co}) + P(\e,\grave{\Theta}_{\cc a})
+ P(\e,\grave{\Theta}_{\cc b}) + P(\grave{\Theta}_{\cc c} ) + P(\grave{\Theta}_{\ccc} ).
\end{equation}
%where $P(\e, \Theta_\co\cup \Theta_\cc\cup \Theta_\ccc )$ represents $P(\e, \Theta_\co\cup \Theta_\cc\cup \Theta_\ccc\mid \text{decoding all links and $\vl_d$ by the given decoder})$, and $P(\e,\grave{\Theta}_i) = P(\e(\D),\Theta_i(\B))$.
%where $P(\e,\grave{\Theta}_i) = P(\e \text{ when $\D$ is used to decode $\vl_d$}, \Theta_i \text{ when the BDD is used to decode the links})$.
\end{lemma}
\begin{IEEEproof}
%First, we note that the event probability of 1bSF can be upper bounded by
%\begin{eqnarray}
%P(\Theta_\co,\Theta_\cc,\Theta_\ccc) & = & P(\e_{12})(1-P(\e^c_{1r})P(\e^c_{2r})P(\e^c_{r2})) \nonumber \\
%&\leq & \stackrel{_\frown}{P}_{12}(1-(1-\stackrel{_\frown}{P}_{1r})(1-\stackrel{_\frown}{P}_{2r})(1-\stackrel{_\frown}{P}_{r2})) \nonumber \\
%& = & \stackrel{_\frown}{P}_{12}(1-\stackrel{_\frown}{P}_{1r})(1-\stackrel{_\frown}{P}_{2r})\stackrel{_\frown}{P}_{r2} + \stackrel{_\frown}{P}_{12}\stackrel{_\frown}{P}_{1r}(1-\stackrel{_\frown}{P}_{2r})(1-\stackrel{_\frown}{P}_{r2}) \nonumber \\
%&& +  \stackrel{_\frown}{P}_{12}(1-\stackrel{_\frown}{P}_{1r})\stackrel{_\frown}{P}_{2r}(1-\stackrel{_\frown}{P}_{r2})  + \stackrel{_\frown}{P}_{12}\stackrel{_\frown}{P}_{1r} \stackrel{_\frown}{P}_{2r}(1-\stackrel{_\frown}{P}_{r2}) + \nonumber \\
%&& +  \stackrel{_\frown}{P}_{12}(1-(1-\stackrel{_\frown}{P}_{1r} )(1-\stackrel{_\frown}{P}_{2r}))\stackrel{_\frown}{P}_{r2}
%\nonumber \\
%&=& P(\grave{\Theta}_{\co}) + P(\grave{\Theta}_{\cc a}) + P(\grave{\Theta}_{\cc b}) + P(\grave{\Theta}_{\cc c}) + P(\grave{\Theta}_{\ccc}).
%\end{eqnarray}
Recall that no hard decoder $\D$ performs worse than the BDD, so it is clear that $\grave{\e}^c_i \subseteq {\e}^c_i \mbox{ and } {\e}_{i} \subseteq \grave{\e}_{i}$, for all $i$. Therefore,
${\Theta}_{\co} \cup {\Theta}_{\cc} \cup {\Theta}_{\ccc} = ({\e}^c_{1r} \cap {\e}^c_{2r} \cap {\e}^c_{r2})^c \cap {\e}_{12} \subseteq
 (\grave{\e}^c_{1r} \cap \grave{\e}^c_{2r} \cap \grave{\e}^c_{r2})^c \cap \grave{\e}_{12}  =\grave{\Theta}_{\co} \cup \grave{\Theta}_{\cc} \cup \grave{\Theta}_{\ccc} $.
Because $\Theta_\co\cup \Theta_\cc\cup \Theta_\ccc\ \subseteq \grave{\Theta}_{\co}\cup \grave{\Theta}_{\cc}\cup \grave{\Theta}_{\ccc}$, we get
$P(\e, \Theta_\co\cup \Theta_\cc\cup \Theta_\ccc )  \leq  P(\e,\grave{\Theta}_{\co}\cup \grave{\Theta}_{\cc}\cup \grave{\Theta}_{\ccc} )
\leq  P(\e,\grave{\Theta}_{\co}) + P(\e,\grave{\Theta}_{\cc a}) + P(\e,\grave{\Theta}_{\cc b}) + P(\grave{\Theta}_{\cc c} ) + P(\grave{\Theta}_{\ccc} )$,
which proves the lemma.
\end{IEEEproof}
%Except $ P(\e,\grave{\Theta}_{\co}) $, which has been studied previously, the terms of the RHS of the above inequality are analyzed in the following.
\noindent

In the following, we analyze
the terms on the right-hand side of \eqref{eq:e.event.bounds} except for the first one, which is given by \eqref{eq:bler-caseI}.
%
%Recall Case II is the case that the relay node $R$ fails {\color{blue}recovering}
%either $\vx_{1}$ or $\vx_{2}$, or both, but the forwarded message
%$\vx_{r}$ from the relay node can be recovered without error at $S2$.
%Three situations will lead to decoding error at $R$, which are 1)
%$\hat{\vx}_{1r}\neq\vx_{1}$ but $\hat{\vx}_{2r}=\vx_{2}$, 2) $\hat{\vx}_{1r}=\vx_{1}$
%but $\hat{\vx}_{2r}\neq\vx_{2}$, and 3) $\hat{\vx}_{1r}\neq\vx_{1}$
%and $\hat{\vx}_{2r}\neq\vx_{2}$. The performance for these situations
%will be studied separately in the following. Let $\mathcal{L}=Rel(\vx_{r})=Rel(\hat{\vx}_{r})$,
%then the combined signal sent to the hard decoder is rewritten as
%\begin{equation}
%L(\vy_{12},\tilde{\vy}_{r2})=4h\sqrt{E}\vy_{12}+\mathfrak{D}(\tilde{\vy}_{r2})\mathcal{L}.
%\end{equation}
%
\subsubsection{$P(\e,\grave{\Theta}_{\cc a})$}
We obtain an upper bound of $P(\e,\grave{\Theta}_{\cc a})=
P(\grave{\Theta}_{\cc a})-P({\e}^c,\grave{\Theta}_{\cc a})$ by deriving a lower bound of $P({\e}^c,\grave{\Theta}_{\cc a})$ as follows.
\begin{thm} \label{thm:barP_caseII.1}
Given a hard decoder of SPF $W_k$, $P({\e}^c,\grave{\Theta}_{\cc a})$ can be lower bounded by
\begin{eqnarray}
&& \stackrel{_\smile}{P}({\e}^c,\grave{\Theta}_{\cc a}) \cr
& \triangleq &(1-\stackrel{_\frown}{P}_{2r})(1-\stackrel{_\frown}{P}_{r2})\sum_{k=2t+1}^{n-1}\beta_{k}(p_{1r}) \sum_{m=\max(0,t+1-(n-k))}^{t} \cr
&& \makebox[-1cm]{}\times \sum_{i=t+1}^{m+(n-k)}\binom{k}{m}\binom{n-k}{i-m}p_{\cc}^{n-k-(i-m)}p_{12}^{m}  \sum_{g=0}^{t-m}\sum_{j=0}^{g}\binom{k-m}{j}
\dot{p}_{\cc}^{j}\dot{q}_{\cc}^{k-m-j}  \binom{i-m}{g-j}\ddot{p}_{\cc}^{g-j}\ddot{q}_{\cc}^{i-m-(g-j)}, \label{eq:lb-CaseII.1}
\end{eqnarray}
where $\beta_{k}(p_{1r})$ is given by \eqref{eq:de-BLER-lower}, and
\begin{equation}\label{eq:noe_yrk_group}
\begin{aligned}
 p_{\cc}   =  1-p_{12}, \quad
 \dot{p}_{\cc}  =  Q(-\sqrt{2h^{2}E})-Q\left(\sqrt{2E}(\mathcal{L}/4hE-h)\right), \\
\dot{q}_{\cc}  =  Q\left(\sqrt{2E}(\mathcal{L}/4hE-h)\right), \quad
 \ddot{p}_{\cc}  =  Q\left(\sqrt{2E}(\mathcal{L}/4hE+h)\right), \\
\ddot{q}_{\cc}   =   Q(\sqrt{2h^{2}E})-Q\left(\sqrt{2E}(\mathcal{L}/4hE+h)\right).
\end{aligned}
\end{equation}
\end{thm}
\begin{proof}
%The proof of this theorem is
%%similar to that of Theorem~\ref{thm:barP_caseI}, and is
%outlined in Appendix~A.
%########   figure 3 ######
%####################
%\begin{figure}
%\center\includegraphics[scale=0.4]{caseII}
%\caption{\label{fig:SF-caseII.1}The error position alignment between the received
%signals from the $R$-$S2$ link and the $S1$-$S2$ link. }
%\end{figure}
Using the inequality  $\stackrel{_\frown}{P}_{1r}\geq \stackrel{_\smile}{P}_{1r}=\sum_{k=2t+1}^{n}\beta_{k}(p_{1r}) $ of Lemma \ref{lemma:lowerbd.of.UBA}, we have
\begin{equation}\label{eq:ineq-CaseII.1}
P({\e}^c,\grave{\Theta}_{\cc a})  \geq   (1-\stackrel{_\frown}{P}_{2r})(1-\stackrel{_\frown}{P}_{r2})\sum_{k=2t+1}^{n-1}\beta_{k}(p_{1r})
 P({\e}^c,d(\vy_{12})\geq t+1 \mid d(\hat{\vx}_{1r})=k,\grave{\e}^c_{2r},\grave{\e}^c_{r2}).
\end{equation}
Note that the term $k=n$ is omitted at \eqref{eq:ineq-CaseII.1} because $P({\e}^c,\grave{\e}_{12}\mid d(\hat{\vx}_{1r})=n,\grave{\e}^c_{2r},\grave{\e}^c_{r2})=0$.
The detailed derivation of \eqref{eq:lb-CaseII.1} from \eqref{eq:ineq-CaseII.1} which proves the theorem is given in Appendix~\ref{app:proof:thm:barP_caseII.1}.
\end{proof}
Using $P(\grave{\Theta}_{\cc a})=\stackrel{_\frown}{P}_{12}\stackrel{_\frown}{P}_{1r}(1-\stackrel{_\frown}{P}_{2r})(1-\stackrel{_\frown}{P}_{r2})$ and Theorem~\ref{thm:barP_caseII.1}, we conclude that $P(\e,\grave{\Theta}_{\cc a})$ can be upper bounded as follows
\begin{equation}\label{eq:PB-caseII.1}
\stackrel{_\frown}{P}(\e,\grave{\Theta}_{\cc a}) \triangleq \stackrel{_\frown}{P}_{12}\stackrel{_\frown}{P}_{1r}(1-\stackrel{_\frown}{P}_{2r})(1-\stackrel{_\frown}{P}_{r2})-\stackrel{_\smile}{P}({\e}^c,\grave{\Theta}_{\cc a}).
\end{equation}

Unlike $P(\grave{\e}^c,\grave{\Theta}_{\co})$ in \eqref{eq:lb-noe-CaseI}, $\stackrel{_\smile}{P}({\e}^c,\grave{\Theta}_{\cc a})$ in Theorem~\ref{thm:barP_caseII.1} is related with the SPF $W_k$.

\subsubsection{$P(\e,\grave{\Theta}_{\cc b})$ }
Similarly, we upper bound
$
P(\e,\grave{\Theta}_{\cc b}) = P(\grave{\Theta}_{\cc b}) - P({\e}^c,\grave{\Theta}_{\cc b}) $ by finding a lower bound of $P({\e}^c,\grave{\Theta}_{\cc b})$ as shown below.
\begin{cor}
Given a hard decoder with SPF $W_k$, $P({\e}^c,\grave{\Theta}_{\cc b})$ can be lower bounded by \eqref{eq:lb-caseII.2}.
\begin{eqnarray}\label{eq:lb-caseII.2}
\stackrel{_\smile}{P}({\e}^c,\grave{\Theta}_{\cc b})
& \triangleq &(1-\stackrel{_\frown}{P}_{1r})(1-\stackrel{_\frown}{P}_{r2})\sum_{k=2t+1}^{n-1}\beta_{k}(p_{2r}) \sum_{m=\max(0,t+1-(n-k))}^{t}\sum_{i=t+1}^{m+(n-k)}\binom{k}{m}\binom{n-k}{i-m} \nonumber \\
&& \makebox[0.5cm]{}\times p_{\cc}^{n-k-(i-m)}p_{12}^{m}\sum_{g=0}^{t-m}\sum_{j=0}^{g}\binom{k-m}{j}
\dot{p}_{\cc}^{j}\dot{q}_{\cc}^{k-m-j}\binom{i-m}{g-j}\ddot{p}_{\cc}^{g-j}\ddot{q}_{\cc}^{i-m-(g-j)}.
\end{eqnarray}
\end{cor}
\begin{proof}
This corollary is proven through a direct use of Theorem \ref{thm:barP_caseII.1} by exchanging the roles of $\e_{1r}$ and $\e_{2r}$.
%the $S1$-$R$ link and the $S2$-$R$ link.
\end{proof}

Thus, by the corollary, $P(\e,\grave{\Theta}_{\cc b})$ is upper bounded by
\begin{equation}\label{eq:PB-caseII.2}
\stackrel{_\frown}{P}(\e,\grave{\Theta}_{\cc b}) \triangleq \stackrel{_\frown}{P}_{12}(1-\stackrel{_\frown}{P}_{1r})\stackrel{_\frown}{P}_{2r}(1-\stackrel{_\frown}{P}_{r2})-\stackrel{_\smile}{P}({\e}^c,\grave{\Theta}_{\cc b}).
\end{equation}

\subsubsection{Upper bound for $P_\text{TW-1bSF}^{(\vx_1)}$}
%Finally, recall
%\begin{equation}
%P(\grave{\Theta}_{\cc c})= \stackrel{_\frown}{P}_{12}\stackrel{_\frown}{P}_{1r}\stackrel{_\frown}{P}_{2r}(1-\stackrel{_\frown}{P}_{r2})
%\label{eq:PB-caseII.3}
%\end{equation}
%and
%\begin{equation}
%P(\grave{\Theta}_{\ccc}) = \stackrel{_\frown}{P}_{12}(1-(1-\stackrel{_\frown}{P}_{1r})(1-\stackrel{_\frown}{P}_{2r}))\stackrel{_\frown}{P}_{r2}.
%\label{eq:PB-caseIII}
%\end{equation}
Substituting \eqref{eq:lb-noe-CaseI} into \eqref{eq:bler-caseI}, \eqref{eq:lb-CaseII.1} into \eqref{eq:PB-caseII.1}, and \eqref{eq:lb-caseII.2} into \eqref{eq:PB-caseII.2}, gives expressions for $P(\grave{\e},\grave{\Theta}_{\co})$, $\stackrel{_\frown}{P}(\e,\grave{\Theta}_{\cc a})$ and $\stackrel{_\frown}{P}(\e,\grave{\Theta}_{\cc b})$, respectively. Then, by substituting these upper bounds, $P(\grave{\Theta}_{\cc c})= \stackrel{_\frown}{P}_{12}\stackrel{_\frown}{P}_{1r}\stackrel{_\frown}{P}_{2r}(1-\stackrel{_\frown}{P}_{r2})$ and $P(\grave{\Theta}_{\ccc}) = \stackrel{_\frown}{P}_{12}(1-(1-\stackrel{_\frown}{P}_{1r})(1-\stackrel{_\frown}{P}_{2r}))\stackrel{_\frown}{P}_{r2}$ into
\eqref{eq:e.event.bounds}, we obtain an upper bound for $P(\e,\Theta_\co\cup\Theta_\cc\cup \Theta_\ccc)$, i.e., $P_\text{TW-1bSF}^{(\vx_1)}$:
%
%\begin{eqnarray}
%\stackrel{_\frown}{P}^{(\vx_1)}_{\text{1bsf}}& =& \stackrel{_\frown}{P}_{12}(1-\stackrel{_\frown}{P}_{1r})(1-\stackrel{_\frown}{P}_{2r})\stackrel{_\frown}{P}_{r2} -\stackrel{_\smile}{P}({\e}^c,\grave{\Theta}_{\co}) + \stackrel{_\frown}{P}_{12}\stackrel{_\frown}{P}_{1r}(1-\stackrel{_\frown}{P}_{2r})(1-\stackrel{_\frown}{P}_{r2})
%\nonumber \\
%&  & - \stackrel{_\smile}{P}({\e}^c,\grave{\Theta}_{\cc a})
% + \stackrel{_\frown}{P}_{12}(1-\stackrel{_\frown}{P}_{1r})\stackrel{_\frown}{P}_{2r}(1-\stackrel{_\frown}{P}_{r2}) -  \stackrel{_\smile}{P}({\e}^c,\grave{\Theta}_{\cc b}) \nonumber \\
%&& + \stackrel{_\frown}{P}_{12}\stackrel{_\frown}{P}_{1r}\stackrel{_\frown}{P}_{2r}(1-\stackrel{_\frown}{P}_{r2})
% + \stackrel{_\frown}{P}_{12}(1-(1-\stackrel{_\frown}{P}_{1r})(1-\stackrel{_\frown}{P}_{2r}))\stackrel{_\frown}{P}_{r2}.
%\end{eqnarray}
\begin{equation}\label{eq:up-TW-1bSF}
\stackrel{_\frown}{P}^{(\vx_1)}_{\text{TW-1bSF}}  \triangleq  P(\grave{\e},\grave{\Theta}_{\co}) + \stackrel{_\frown}{P}(\e,\grave{\Theta}_{\cc a})
  + \stackrel{_\frown}{P}(\e,\grave{\Theta}_{\cc b}) + P(\grave{\Theta}_{\cc c}) + P(\grave{\Theta}_{\ccc}).
\end{equation}

\begin{rem}[\emph{Performance of Perfect Codes}]
For perfect codes, the sphere decoder is exactly the optimal hard decoder. It has $W_{2t+1}=A_{2t+1}\binom{2t+1}{t}=\binom{n}{t+1}$ and $W_k=0$ for all $k>2t+1$. Then the lower bound given by \eqref{eq:BLER-lower} and the upper bound given by \eqref{eq:BLER-upper} are equal, and also equal to the lower bound \eqref{eq:de-BLER-lower} of this paper. Thus, the upper bound of TW-SDF given by \eqref{eq:bler-SDF} is exact. Furthermore, the upper bound derived for TW-1bSF is asymptotically tight because the results for $P(\e, \grave{\Theta}_{\co})$, $P(\e, \grave{\Theta}_{\cc a})$ and $P(\e, \grave{\Theta}_{\cc b})$ are exact, whereas $P(\grave{\Theta}_{\cc c})$ and $P(\grave{\Theta}_{\ccc})$ are asymptotically negligible.
\end{rem}

\section{Asymptotic performance analysis of TW-SDF and TW-1bSF}
%
%Denote $\sim_{E}$ the asymptotic equality as $E\rightarrow\infty$.
%For simple notation, the subscript $E$ is normally omitted.
The analytical comparison of the BLER performance of the TW-1bSF and the TW-SDF protocols is prohibitive because of the fairly long BLER expressions. To gain insight, we derive and compare their asymptotic performance as $E\rightarrow \infty$ and $E_r\rightarrow \infty$.
%In principle, $P^{(\vx_1)}_{\text{TW-1bSF}}$ has different asymptotic forms with respect to different asymptotic behavior among $hE$, $h_1E$, $h_2E$ and $h_2E_r$, and likewise $P^{(\vx_1)}_{\text{TW-SDF}}$ has different asymptotic forms. For example, we can study the asymptotic performance of them with respect to a fixed $E$ and $E_r\rightarrow \infty$, or the opposite. To shed light on the advantage of the 1bSF protocol over the SDF protocol with least effort,
%we compare the asymptotic performance of them with respect to $E\rightarrow \infty$ and $E_r\rightarrow \infty$.
In the following analysis, $f_{A}(E,E_r)$ and $f_{B}(E,E_r)$ are said to be asymptotically equal with respect to $E$ and $E_r$, denoted as $f_{A}(E,E_r)\sim f_{B}(E,E_r)$, if $\log_{E\rightarrow\infty,E_r\rightarrow\infty}f_{A}(E,E_r)/f_{B}(E,E_r)=1$.
%We say $f_A(E)\gtrsim f_C(E)$ if $f_A(E)\geq f_B(E)\sim f_C(E)$, and $\lesssim$ is defined in the similar way.
We first present a lemma that will be used in the asymptotic analysis afterwards.
\begin{lemma}\label{lemma:asym-ratio}
Let $z_1\sim\mathcal{N}(\mu_1\sqrt{E},1)$ and $z_2\sim\mathcal{N}(\mu_2\sqrt{E},1)$
be two independent Gaussian random variables. If $\Lambda_1=\Pr(z_1+z_2>0,z_1>0,z_2<0)$
and $\Lambda_2=\Pr(z_1+z_2<0,z_1>0,z_2<0)$, then $\Lambda_1/\Lambda_2\rightarrow \infty$ as $E \rightarrow \infty$.
\end{lemma}
\begin{IEEEproof}
A detailed proof is given in Appendix~\ref{app:proof:lem:asym-ratio}.
\end{IEEEproof}

\subsection{Asymptotic performance of TW-SDF}
When $E\rightarrow \infty$ and $E_r\rightarrow \infty$, i.e., $p_i\rightarrow 0$, the upper bound \eqref{eq:BLER-upper}
converges to $\stackrel{_\frown}{P}_i\sim \tbinom{n}{t+1}p_i^{t+1}$.
We start by analyzing $P(\grave{\e},\grave{\Theta}_{\co})$ in \eqref{eq:bler-caseI}.
%So, the BLER approximations of the $i=\{s,1,2,r\}$ links have the asymptotic expressions as following:
%\begin{eqnarray}
%\stackrel{_\frown}{P}_{12} & \sim & \binom{n}{t+1}p_{12}^{t+1}, \label{eq:asym-ps}\\
%\stackrel{_\smile}{P}'_{1r} & \sim & A_{2t+1}\binom{2t+1}{t}p_{1r}^{t+1}, \label{eq:asym-p1}\\
%\stackrel{_\smile}{P}'_{2r} & \sim & A_{2t+1}\binom{2t+1}{t}p_{2r}^{t+1}, \label{eq:asym-p2}\\
%PB_r^{(l)} & \sim & A_{2t+1}\binom{2t+1}{t}p_{r2}^{t+1}. \label{eq:asym-pr}
%\end{eqnarray}

Case I): Since $\stackrel{_\frown}{P}_{12}(1-\stackrel{_\frown}{P}_{1r})(1-\stackrel{_\frown}{P}_{2r}) \stackrel{_\frown}{P}_{r2}  \sim \stackrel{_\frown}{P}_{12}\stackrel{_\frown}{P}_{r2} \\ \sim
\binom{n}{t+1}p_{12}^{t+1}\binom{n}{t+1}p_{r2}^{t+1}$, asymptotically
$P(\grave{\e},\grave{\Theta}_{\co})\sim \binom{n}{t+1}p_{12}^{t+1}\binom{n}{t+1}p_{r2}^{t+1}-P(\grave{\e}^c,\grave{\Theta}_{\co})$.
Regarding $P(\grave{\e}^c,\grave{\Theta}_{\co})$ given by \eqref{eq:lb-noe-CaseI}, it is clear from Lemma~\ref{lemma:asym-ratio} that $\dot{q}_{\co}/\dot{p}_{\co}\rightarrow \infty$ and  $\ddot{q}_{\co}/\ddot{p}_{\co}\rightarrow \infty$, and hence $\dot{q}_{\co}\sim p_{r2}(1-p_{12})\sim p_{r2}$ and $\ddot{q}_{\co}\sim p_{12}(1-p_{r2}) \sim p_{12}$. Thus, $P(\grave{\e}^c,\grave{\Theta}_{\co})$ is dominated by the terms in which $\dot{p}_{\co}^j=\ddot{p}_{\co}^{g-j}=1$, i.e., $j=g=0$, $i = t+1 $ and $k  = t+1$. Therefore
\begin{IEEEeqnarray}{rCl}
 P(\grave{\e}^c,\grave{\Theta}_{\co})
& \sim &  \binom{n}{t+1}
\sum_{m=0}^{t}\binom{t+1}{m}\binom{n-(t+1)}{t+1-m}
p_{r2}^{m}p_{12}^{m}\dot{q}_{\co}^{t+1-m}\ddot{q}_{\co}^{t+1-m} \nonumber\\
 & \sim & \binom{n}{t+1}p_{r2}^{t+1}p_{12}^{t+1}
\sum_{m=0}^{t}\binom{t+1}{m}\binom{n-(t+1)}{t+1-m}.
\end{IEEEeqnarray}
Thus,
\begin{IEEEeqnarray}{rCl}\label{eq:asym-caseI}
P(\grave{\e},\grave{\Theta}_{\co}) & \sim & \binom{n}{t+1}^2 p_{12}^{t+1}p_{r2}{}^{t+1}
- \binom{n}{t+1} p_{r2}^{t+1}p_{12}^{t+1}  \sum_{m=0}^{t}\binom{t+1}{m}\binom{n-(t+1)}{t+1-m} \nonumber\\
&\sim &  p_{12}^{t+1}p_{r2}^{t+1}\binom{n}{t+1}.
\end{IEEEeqnarray}
To derive \eqref{eq:asym-caseI}, we used the fact that $
\tbinom{n}{t+1}-\sum_{r=0}^{t}\tbinom{
t+1}{r}\tbinom{
n-(t+1)}{
t+1-r}=\tbinom{
t+1}{t+1}=1
$.

Finally, by substituting this asymptotic form and
$P(\grave{\Theta}_{\cc}\cup \grave{\Theta}_{\ccc})
\sim \stackrel{_\frown}{P}_{12}(\stackrel{_\frown}{P}_{1r}+\stackrel{_\frown}{P}_{2r})
$ into \eqref{eq:bler-SDF},
we obtain
\begin{equation}
\stackrel{_\frown}{P}^{(\vx_1)}_{\text{TW-SDF}}
  \sim   \stackrel{_\frown}{P}^{(\vx_1)}_{\text{TW-SDF,asym}}
 \triangleq  \binom{n}{t+1} p_{12}^{t+1}p_{r2}{}^{t+1}
+\binom{n}{t+1}^2[p_{1r}^{t+1}+p_{2r}^{t+1}]p_{12}^{t+1}.
\label{eq:asym-SDF}
\end{equation}
%\begin{eqnarray}
%\stackrel{_\frown}{P}^{(\vx_1)}_{\text{TW-SDF}}
%& \sim & \stackrel{_\frown}{P}^{(\vx_1)}_{\text{TW-SDF,asym}} \nonumber \\
%& \triangleq& \binom{n}{t+1} p_{12}^{t+1}p_{r2}{}^{t+1} +\binom{n}{t+1}^2[p_{1r}^{t+1}+p_{2r}^{t+1}]p_{12}^{t+1}.
%\label{eq:asym-SDF}
%\end{eqnarray}

\subsection{Asymptotic performance of TW-1bSF}
Among the terms of $\stackrel{_\frown}{P}^{(\vx_1)}_{\text{TW-1bSF}}$ in \eqref{eq:up-TW-1bSF}, we have derived an asymptotic expression for $P(\grave{\e},\grave{\Theta}_{\co})$ in \eqref{eq:asym-caseI}.
In the following, we derive the remaining terms for Case II and Case III.

Case II): In $\stackrel{_\frown}{P}(\e,\grave{\Theta}_{\cc a})$ given by \eqref{eq:lb-CaseII.1}, $\dot{q}_{\cc}$ is the probability that an erroneous bit $i$ from the relay with $\hat{x}_{1r,i}=-1$ can be corrected by $y_{12,i}>0$, as shown in Fig.~\ref{fig:SF-caseII.1}. To optimize the performance of TW-1bSF, $\dot{q}_{\cc}$ should be maximized. Therefore it is desirable that $\dot{q}_{\cc}>\dot{p}_{\cc}$, and we want $\tfrac{\mathcal{L}}{4hE}\sim \tfrac{\min(h_{1}^{2},h_{2}^{2})}{4h}<h$ to hold. Here, to derive the asymptotic form of $\mathcal{L}$ we have used the approximation $Q(x)\sim \tfrac{1}{2}\exp(-\tfrac{1}{2}x^2)$ and $1-Q(x)\sim 1$.

If the desired condition is satisfied, then $\dot{q}_{\cc}/ \dot{p}_{\cc}\rightarrow \infty$ as $E\rightarrow\infty$, and therefore $\dot{q}_{\cc}\sim1$. Note that $\tfrac{\mathcal{L}}{4hE}\sim \tfrac{\min(h_{1}^{2},h_{2}^{2})}{4h}<h$ holds when $h=h_{1}=h_{2}=1$, for instance.
Furthermore, it is clear that $\ddot{q}_{\cc}/\ddot{p}_{\cc}\rightarrow \infty$, and so $\ddot{q}_{\cc}\sim Q(\sqrt{2h^{2}E})=p_{12}$.
Based on the asymptotic analysis for $\dot{q}_{\cc}$, $\dot{p}_{\cc}$, $\ddot{q}_{\cc}$ and $\ddot{p}_{\cc}$, we conclude that in $\stackrel{_\frown}{P}(\e,\grave{\Theta}_{\cc a})$ the dominant events happen when $i=t+1$,
$g=0$, and $j=0$, for all $k$ with $W_k\geq 0$, so
$P(d(\vl_d)\leq t,d(\vy_{12})\geq t+1 \mid d(\hat{\vx}_{1r})=k,\mathfrak{D}(\tilde{\vy}_{r2})=\hat{\vx}_{1r}) \sim
\sum_{m=0}^{t}\binom{k}{m}\binom{n-k}{t+1-m}p_{12}^{t+1}\dot{q}_{\cc}^{k-m}$ in \eqref{eq:noe_yrk}.
Therefore, by substituting the asymptotic form of \eqref{eq:noe_yrk} and $\beta_{k}(p_{1r}) \sim W_k p_{1r}^{t+1}$ into \eqref{eq:lb-CaseII.1}, and then into \eqref{eq:PB-caseII.1}, we get
\begin{IEEEeqnarray}{rCl}\label{eq:asym-caseII.1}
 && \stackrel{_\frown}{P}(\e,\grave{\Theta}_{\cc a}) \nonumber \\
%& \sim & \stackrel{_\frown}{P}_{12}\stackrel{_\frown}{P}_{1r}-\sum\nolimits_{k=2t+1}^{n}\beta_{k}(p_{1r}) P(d(\vl_d)\leq t,d(\vy_{12})\geq t+1 \mid d(\hat{\vx}_{1r})=k,\mathfrak{D}(\tilde{\vy}_{r2})=\hat{\vx}_{1r})  \nonumber \\
 & \sim & \binom{n}{t+1}^2 p_{12}^{t+1}p_{1r}^{t+1}
 - \sum_{k=2t+1}^{n} W_k p_{1r}^{t+1}p_{12}^{t+1}\sum_{m=0}^{t}\binom{k}{m}\binom{n-k}{t+1-m}
 \dot{q}_{\cc}^{k-m} \nonumber \\
&\sim &p_{1r}^{t+1}p_{12}^{t+1} \left[\binom{n}{t+1}^2
- \sum_{k=2t+1}^{n} W_k  \left( \binom{n}{t+1}-\binom{k}{t+1} \right)
\right].
\end{IEEEeqnarray}
For the derivation of the last step, we have used the fact that $\sum_{m=0}^{t}\binom{k}{m}\binom{n-k}{t+1-m} = \binom{n}{t+1}-\binom{k}{t+1}$.

By replacing $p_{1r}$ with $p_{2r}$ in the asymptotic equation \eqref{eq:asym-caseII.1}, we have
\begin{equation}
\stackrel{_\frown}{P}(\e,\grave{\Theta}_{\cc b})
\sim  p_{2r}^{t+1}p_{12}^{t+1}
\left[\binom{n}{t+1}^2
- \sum_{k=2t+1}^{n} W_k  \left( \binom{n}{t+1}-\binom{k}{t+1} \right)
\right].
 \label{eq:asym-caseII.2}
\end{equation}

From $P(\grave{\Theta}_{\cc c})= \stackrel{_\frown}{P}_{12}\stackrel{_\frown}{P}_{1r}\stackrel{_\frown}{P}_{2r}(1-\stackrel{_\frown}{P}_{r2})$, it is clear that
$P(\grave{\Theta}_{\cc c})\sim \binom{n}{t+1}^{3} (p_{1r}p_{2r})^{t+1}p_{12}^{t+1}$.
By comparing $P(\grave{\Theta}_{\cc c})\sim O((p_{1r}p_{2r})^{t+1}p_{12}^{t+1})$ with $\stackrel{_\frown}{P}(\e,\grave{\Theta}_{\cc a})$ and $\stackrel{_\frown}{P}(\e,\grave{\Theta}_{\cc b})$, which are $O(p_{12}^{t+1}p_{1r}^{t+1})$ and $O(p_{12}^{t+1}p_{2r}^{t+1})$, we find that $P(\grave{\Theta}_{\cc c})$ is asymptotically negligible.

Case III): Using similar arguments as applied to derive the asymptotic form of $P(\grave{\Theta}_{\cc c})$, we also find that $P(\grave{\Theta}_{\ccc})$ is asymptotically negligible compared with $\stackrel{_\frown}{P}(\e,\grave{\Theta}_{\co})$, $\stackrel{_\frown}{P}(\e,\grave{\Theta}_{\cc a})$ and $\stackrel{_\frown}{P}(\e,\grave{\Theta}_{\cc b})$.

Finally, by substituting \eqref{eq:asym-caseI}, \eqref{eq:asym-caseII.1} and \eqref{eq:asym-caseII.2} into \eqref{eq:up-TW-1bSF}, we obtain the following asymptotic expression for $\stackrel{_\frown}{P}^{(\vx_1)}_{\text{TW-1bsf}}$
\begin{eqnarray}
&&\stackrel{_\frown}{P}^{(\vx_1)}_{\text{TW-1bSF,asym}} \cr
& \triangleq  & p_{12}^{t+1}p_{r2}^{t+1}\binom{
n}{t+1}+(p_{1r}^{t+1}+p_{2r}^{t+1})p_{12}^{t+1}
 \left[\binom{n}{t+1}^2
- \sum_{k=2t+1}^{n} W_k  \left( \binom{n}{t+1}-\binom{k}{t+1} \right)
\right].
\label{eq:asym-SF}
\end{eqnarray}

%From \eqref{eq:SDF-e} and \eqref{eq:SDF-c}, we have

\subsection{Asymptotic Performance Comparison}
Comparing the asymptotic performance of TW-SDF and TW-1bSF given
by  \eqref{eq:asym-SDF} and \eqref{eq:asym-SF}, respectively, we note that they have the same first term, but the second term of  \eqref{eq:asym-SF}
is smaller than the second term of \eqref{eq:asym-SDF}. This observation means
that TW-1bSF has a performance gain over TW-SDF if the error events of Case II are dominant
over the error events of Case I, which happens if $p_{r2}\leq p_{1r}$ or $p_{r2}\leq  p_{2r}$.
Now, suppose that $E_{r}=E$ and $h=h_{1}=h_{2}=1$, which satisfy both conditions $p_{r2} \leq p_{2r}$ and $\dot{q}_{\cc}\sim1$. Then the asymptotic performance gain of the TW-1bSF protocol over the TW-SDF protocol is given by
\begin{equation}
\frac{\stackrel{_\frown}{P}^{(\vx_1)}_{\text{TW-SDF,asym}}}{\stackrel{_\frown}{P}^{(\vx_1)}_{\text{TW-1bSF,asym}}}
=\frac{\binom{n}{t+1}^2}{\binom{n}{t+1}^2
- \sum_{k=2t+1}^{n} W_k \left( \binom{n}{t+1}-\binom{k}{t+1} \right)
}.
\label{eq:asym-gain}
\end{equation}

If perfect codes like Hamming codes are applied, then $W_{2t+1} = \binom{n}{t+1}$ and $W_k = 0$ for $k=0$ and $k>2t+1$. Then the asymptotic gain \eqref{eq:asym-gain} can be further simplified to
\begin{equation}
\frac{\stackrel{_\frown}{P}^{(\vx_1)}_{\text{TW-SDF,pc,asym}}}{\stackrel{_\frown}{P}^{(\vx_1)}_{\text{TW-1bSF,pc,asym}}}= \binom{n}{t+1}\left/\binom{2t+1}{t+1}\right. .
\end{equation}
This asymptotic ratio for perfect codes shows that the performance
gain increases with the codeword length $n$.
%
%We can further simplify the asymptotic gain by noting that when $n\gg k$, $\binom{n}{t+1}-\binom{k}{t+1} \approx \binom{n}{t+1}-\binom{2t+1}{t+1}$. Since for all $k$ with $W_k > 0$, $n\gg k$ is valid for large n, we can approximate $P^{(\vx_1,u)}_{\text{1bSF,asym}}$ as
%\begin{eqnarray}
%P^{(\vx_1,u)}_{\text{1bSF,asym}} & \approx & p_{12}^{t+1}p_{r2}^{t+1}\binom{
%n}{t+1}+(p_{1r}^{t+1}+p_{2r}^{t+1})p_{12}^{t+1}\left[\binom{n}{t+1}^2
%- \sum_{k=2t+1}^{n} W_k  \left( \binom{n}{t+1}-\binom{2t+1}{t+1} \right)\right] \nonumber \\
%&\approx& p_{12}^{t+1}p_{r2}^{t+1}\binom{
%n}{t+1}+(p_{1r}^{t+1}+p_{2r}^{t+1})p_{12}^{t+1}\binom{n}{t+1}\binom{2t+1}{t+1}
%\end{eqnarray}
%And the asymptotic gain is approximate to
%\begin{equation}
%\frac{P^{(\vx_1,u)}_{\text{SDF,asym}}}{P^{(\vx_1,u)}_{\text{1bSF,asym}}}\approx \binom{n}{t+1}/\binom{2t+1}{t+1} .
%\end{equation}
%This approximate asymptotic gain shows that the asymptotic performance
%gain increases as the codeword length. This conclusion will be verified
%at Section V by illustrative result comparison.

\begin{rem}[\emph{Redesign of the Reliability Value}] \label{rem:rel-design} From the asymptotic analysis above, in order to maximize the performance gain of the TW-1bSF protocol over the TW-SDF protocol, two conditions need to be met: (1) $\mathcal{L}/4hE<h$ in order for $\dot{q}_{\cc}\sim1$; (2) $p_{r2}\leq p_{1r}$ or $p_{r2}\leq  p_{2r}$ so that the rate of error events in Case II dominates over the rate of error events in Case I. Condition (2) can be easily satisfied, since $p_{r2}=p_{2r}$ when $E_r = E$. Regarding Condition (1), the original choice for $\mathcal{L}\sim\min(h_{1}^{2}E,h_{2}^{2}E)$ provides no guarantee on whether the condition is met.
For instance, the condition is violated if $h=1$, $h_{1}=h_{2}=2$.
In order to satisfy the condition, the reliability value can be
chosen as \begin{equation}
\mathcal{L}^*=\min\left(\log\frac{1-p_{1r}}{p_{1r}},\log\frac{1-p_{2r}}{p_{2r}},\log\frac{1-p_{12}}{p_{12}}\right).\label{eq:rel-new}\end{equation}
%The connection to $p_{12}$ is because the packet of binary symbols share the same reliability value. In other word, due to the one-bit quantization.
The superiority of this design choice will be illustratively verified in
Section~VI through simulations. On the other hand, the design of $\mathcal{L}^-$ is not as critical as $\mathcal{L}$ since the choice of $\mathcal{L}^-$ does not influence the asymptotic performance.
\end{rem}

\begin{rem}[\emph{Decoding Energy Consumption Comparison}]
Let $\epsilon$ be the energy consumption that is needed to decode a codeword. We first consider the consumption at $S2$. Given that $\vy_{12}$ is in error, in the TW-SDF protocol, if  $R$ decodes the received signals correctly and forwards the network-coded word, $S2$ will consume $3\epsilon$ to decode $\vy_{12}$, $\tilde{\vy}_{r2}$ and $(4h\sqrt{E}\vy_{12}+4h_2\sqrt{E_r}\tilde{\vy}_{r2})$; otherwise, $R$ forwards nothing, and $S2$ merely decodes $\vy_{12}$ consuming $\epsilon$. Therefore, the
decoding energy consumed at $S2$ by the TW-SDF protocol is
\begin{align}
E_{\text{TW-SDF},S2}&= \epsilon P(\e^c_{12})+ 3\epsilon P(\e_{12},\e^c_{1r},\e^c_{2r})
+\epsilon (P(\e_{12})-P(\e_{12},\e^c_{1r},\e^c_{2r})) \cr
&\approx   \epsilon(1-\stackrel{_\frown}{P}_{12})+3\epsilon P(\Theta_\co) + \epsilon P(\Theta_{\cc}\cup\Theta_{\ccc}).
\end{align}
In the TW-1bSF protocol, if $\vy_{12}$ is in error, no matter whether $R$ can decode the received signals correctly, it will forward $\vx_r$, so $S2$ consumes $3\epsilon P(\Theta_{\cc}\cup\Theta_{\ccc})$.
The decoding energy consumed by the TW-1bSF protocol is given by
\begin{equation}
E_{\text{TW-1bSF},S2}\approx \epsilon(1-\stackrel{_\frown}{P}_{12})+3\epsilon P(\Theta_\co) + 3\epsilon P(\Theta_{\cc}\cup\Theta_{\ccc}).
\end{equation}

Due to the symmetry of the model and $\stackrel{_\frown}{P}_{12} =\stackrel{_\frown}{P}_{21}$, the decoding energy consumption at $S1$ is given by $E_{\text{TW-SDF},S1} = E_{\text{TW-SDF},S2}$, and $E_{\text{TW-1bSF},S1} = E_{\text{TW-1bSF},S2}$. As to the relay, no matter whether TW-SDF or TW-1bSF is applied, it decodes the messages from
$S1$ and $S2$, and thus consumes energy $E_{\text{TW-SDF},R} = E_{\text{TW-1bSF},R} = 2\epsilon$. Finally, regarding the whole system including
$S1$, $S2$ and $R$, the decoding energy consumption for TW-SDF and TW-1bSF is given, respectively, by
\begin{align}
E_{\text{TW-SDF}} & =  2\epsilon(2-\stackrel{_\frown}{P}_{12})+6\epsilon P(\Theta_\co) + 2\epsilon P(\Theta_{\cc}\cup\Theta_{\ccc}), \\
E_{\text{TW-1bSF}}& =  2\epsilon(2-\stackrel{_\frown}{P}_{12})+6\epsilon P(\Theta_\co) + 6\epsilon P(\Theta_{\cc}\cup\Theta_{\ccc}).
\end{align}

Since $P(\Theta_{\co})\sim\stackrel{_\frown}{P}_{12}\stackrel{_\frown}{P}_{r2}$ and $P(\Theta_{\cc}\cup\Theta_{\ccc})\sim \stackrel{_\frown}{P}_{12}(\stackrel{_\frown}{P}_{1r}+\stackrel{_\frown}{P}_{2r})$, we have $E_{\text{TW-SDF}}\sim E_{\text{TW-1bSF}}\sim 2\epsilon(2-\stackrel{_\frown}{P}_{12})$, which means that the TW-1bSF protocol consumes almost the same decoding energy as TW-SDF.
\end{rem}

\begin{rem}[\emph{Transmission Energy Consumption Comparison}]
In TW-1bSF, both sources transmit a message with power $E$, while the relay forwards a network-coded message with power $E_r$ during the third timeslot if any of the two direct links fails, i.e., $\e_{12}$ or $\e_{21}$. Therefore, to exchange a pair of messages, the overall transmission energy consumed by TW-1bSF is
\begin{equation}
T_{\text{TW-1bSF}} \approx 2E + E_r(\stackrel{_\frown}{P}_{12}+\stackrel{_\frown}{P}_{21}).
\end{equation}
In order for the relay at TW-SDF to forward, besides $\e_{12}$ or $\e_{21}$ as in TW-1bSF, $\e^c_{1r}$ and $\e^c_{2r}$ should hold. Therefore
\begin{equation}
T_{\text{TW-SDF}} \approx 2E + E_r(\stackrel{_\frown}{P}_{12}+\stackrel{_\frown}{P}_{21})(1-\stackrel{_\frown}{P}_{1r})(1-\stackrel{_\frown}{P}_{2r}).
\end{equation}
Therefore, $T_{\text{TW-SDF}}\sim T_{\text{TW-1bSF}}\sim 2E + 2E_r\stackrel{_\frown}{P}_{12}$,
which means that the TW-1bSF protocol consumes almost the same transmission energy as the TW-SDF protocol.
\end{rem}

\section{Results and Discussion}

The proposed TW-1bSF protocol is simple in terms of signal processing complexity, and hence consumes relatively little power. Hence, it is suitable for use in communication systems with stringent power constraints such
as Bluetooth (IEEE 802.15.1) \cite{Razavi2007}, IEEE 802.15.6 \cite{Kwak2011} and WBAN (IEEE 802.15.6) \cite{Kwak2011}. Therefore, in order to assess the performance
of the TW-1bSF protocol we use channel codes utilized by these standards. As an example of perfect codes, the
(15,11) Hamming code is adopted by Bluetooth \cite{Razavi2007}, and so it is studied in this section. Let $r_{c}$ be the code rate and $E_{b}/N_{0}$
be the SNR of the information bits. Then the SNR of the code bits is $E/N_{0}=r_{c}E_{b}/N_{0}$.
In sensor networks, in general all nodes can switch their mode between acting as a
transmitter and as a relay, so it is natural to assume that $E=E_{r}$.
In the rest of analysis, it is also assumed that $h=h_{1}=h_{2}=1$,
because the channel changes slowly and we can always use power control to compensate for the effects of fading. In this setting, both conditions in Remark~\ref{rem:rel-design} that render TW-1bSF preferable are satisfied.

\begin{figure}
\center\includegraphics[width = 5in]{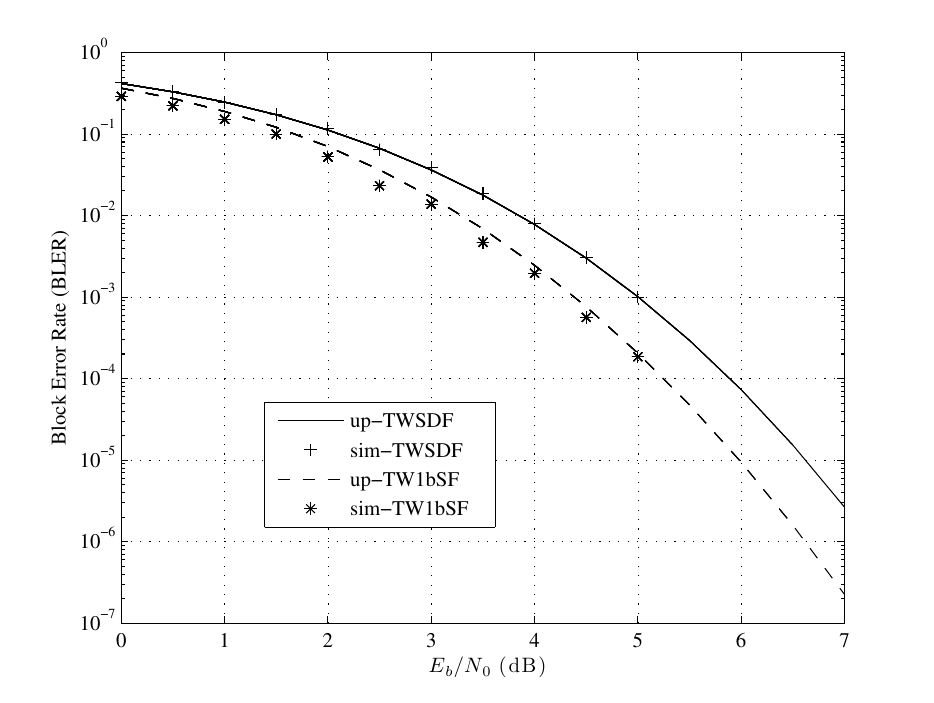}
%\centerline{(a)}
%\center\includegraphics[scale=0.7]{snr0-7berHamingN15}
%\centerline{(b)}
\caption{\label{fig:Hamming-SF}Simulation results vs analytical results for
TW-1bSF and TW-SDF in the two-way network setting for the (15,11) Hamming code. ``sim'' stands for simulation,  whereas ``up'' stands for the derived BLER upper bounds
\eqref{eq:bler-SDF} and \eqref{eq:up-TW-1bSF}.}
\end{figure}

Fig.~\ref{fig:Hamming-SF} shows the simulated
and analytical BLER performance for TW-SDF and TW-1bSF when the (15,11) Hamming code is employed.
It can be seen that the simulation results are
very close to the upper bound derived for TW-1bSF, and
match perfectly with the derived upper bound for TW-SDF. This comparison verifies
the accuracy of the performance analysis presented in Section IV. It also shows that for BLERs ranging from $10^{-2}$
to $10^{-3}$, the SNR gain of TW-1bSF over TW-SDF is around 0.6 dB, which
offers power savings of about $10^{0.6/10}-1 \approx 15\%$.
% #########  figure 4   ########
%#########################
%\begin{figure}
%\center\includegraphics[scale=0.7]{snr0-7BLERHamingN15}
%%\centerline{(a)}
%%\center\includegraphics[scale=0.7]{snr0-7berHamingN15}
%%\centerline{(b)}
%\caption{\label{fig:Hamming-SF}Simulation results vs analytical results of
%1bSF and SDF for the two-way network using (15,11) Hamming code. ``sim'' stands for simulation,  while in (a) ``up'' stands for the derived BLER upper bound, and in (b) ``up'' stands for results calculated by using the approximation $P_b \approx \tfrac{2t+1}{n}P_B$ in which $P_B$ is the analytical BLER upper bound shown in (a).}
%\end{figure}

\setcounter{table}{0}

\begin{table*}
\centering
\caption{Values of the sphere partition function (SPF) $W_k$ for two-error-correcting $(t=2, d_{\min}=5$) BCH codes with code length $n=15,31,63,127,255, \text{and } 511$.}
\label{table:SPF}
\begin{tabular}{|c||c|c|c|c|c|c|}
\hline
$W_{k}$ & $n=15$ & $n=31$ & $n=63$ & $n=127$ & $n=255$ & $n=511$\tabularnewline
\hline
\hline
$k=0$ & 36 & 71 & 170 & 164 & 336 & 362\tabularnewline
\hline
$k=d_{\min}$ & 281 & 2123 & 19316 & 161664 & 1351834 & 10949163\tabularnewline
\hline
$k=d_{\min}+1$ & 89 & 540 & 2557 & 5399 & 19172 & 42719\tabularnewline
\hline
\end{tabular}
\end{table*}

As an extension of Hamming codes, BCH codes are capable of correcting multiple error bits and hence are also widely adopted by small-area
communication systems \cite{Sankarasubramaniam2003,Kwak2011}. For the simulations of this paper, the BCH codewords are generated by a systematic encoder, and the received words are decoded by a Berlekamp-Massey decoder; both the encoder and the decoder are provided by the Matlab Communication Toolbox. It is worthwhile to note that in the Matlab BCH decoder, if a decoding failure happens, then the first $k$ received bits are decoded as the estimate of the transmitted information message if the code rate is $k/n$. The SPF $W_k$ associated with the Matlab encoder/decoder is shown in Table~\ref{table:SPF} for some typical BCH codes. The BLER results obtained for the (127,113) code by simulation and using the analytical upper bound for TW-1bSF and TW-SDF are shown in Fig.~\ref{fig:BCH-BER}.
Again, the simulation results
are tightly upper bounded by the derived bound for TW-SDF, and well bounded by the bound for TW-1bSF,
which is asymptotically tight. The simulation results also show that, when the BLER is in the range between
$10^{-2}$ and $10^{-3}$, the SNR gain attained by TW-1bSF is about 0.8
dB, which corresponds to power savings of $10^{0.8/10}\approx 20\%$. Fig.~\ref{fig:Hamming-SF} and \ref{fig:BCH-BER} indicate that the asymptotic tightness of the derived upper bounds for the TW-SDF and the TW-1bSF protocols holds for any hard decoders with negligible $W_0$. This is because these decoders can be approximated by the BDD, which has $W_0=0$, and by which the  upper bounds are achieved. As can be seen in Table.~\ref{table:SPF}, the Matlab BCH decoder that is employed has very small $W_0$.

\begin{figure}
\center\includegraphics[width = 5in]{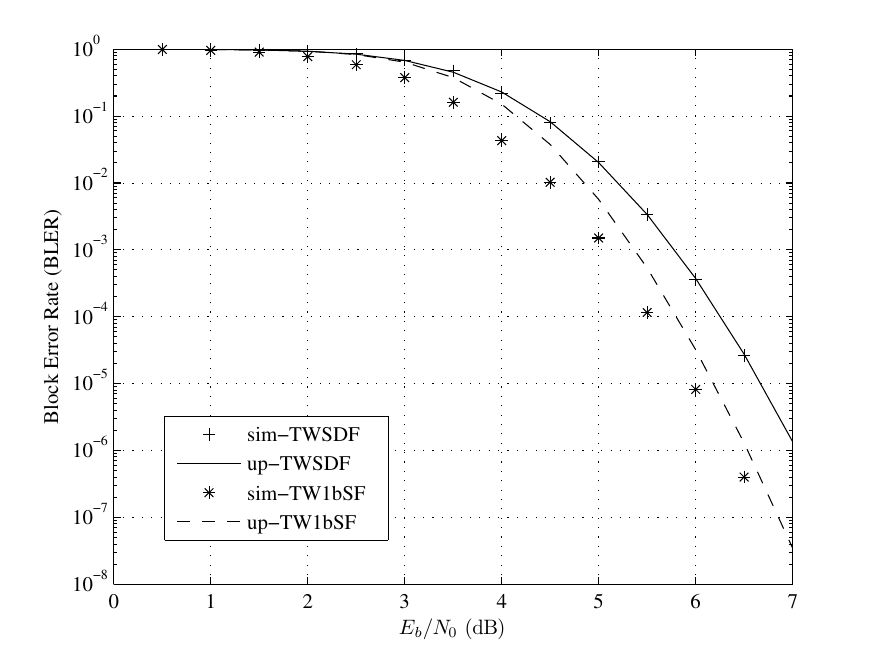}
%\centerline{(a)}
%\center\includegraphics[scale=0.7]{bitBCHN127K113}
%\centerline{(b)}
\caption{\label{fig:BCH-BER}Performance of TW-SDF and TW-1bSF for the two-way network using the (127,113) BCH code.}
\end{figure}

\begin{figure}
\center\includegraphics[width = 5in]{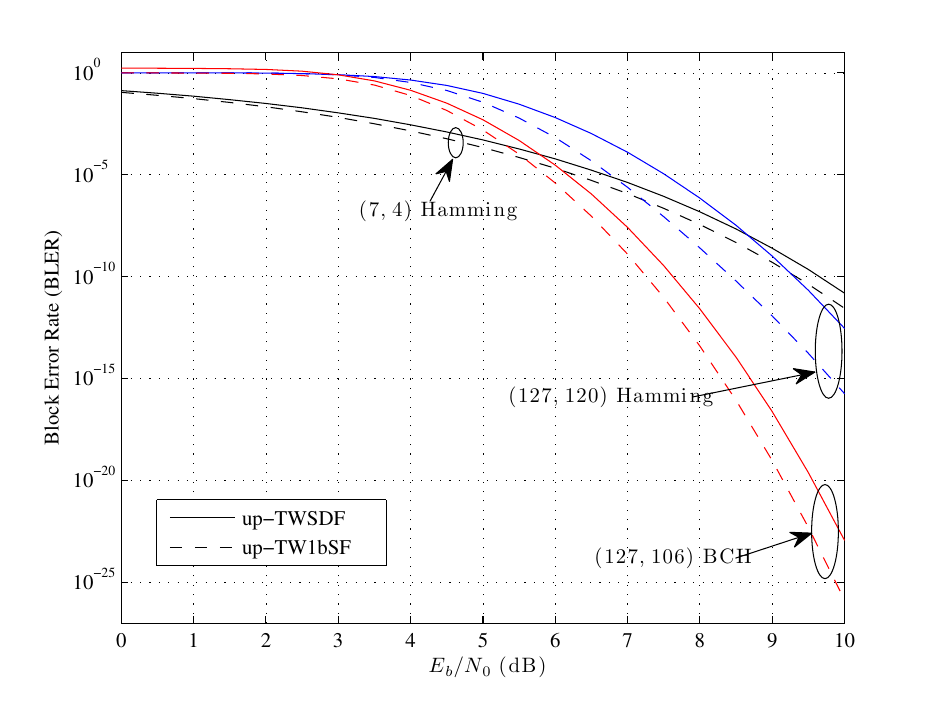}
%\centerline{(a)}
%\center\includegraphics[scale=0.7]{diffCodesBer}
%\centerline{(b)}
\caption{\label{fig:codelength}Performance comparison of TW-SDF and TW-1bSF for different
codes in terms of the block error rate (BLER). ``$(7,4)$ Hamming'' represents the $(7,4)$ Hamming code,
``$(127,120)$ Hamming'' represents the $(127,120)$ Hamming code and ``$(127,106)$ BCH'' represents the $(127,106)$ BCH code.}
\end{figure}

\begin{figure}
\center\includegraphics[width = 5in]{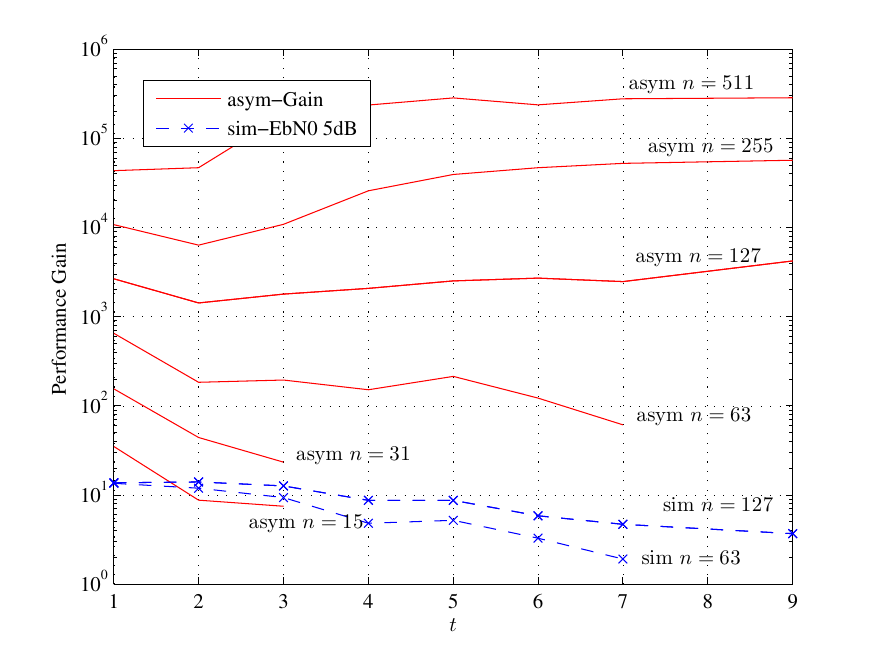}
\caption{\label{fig:asym-gain}Asymptotic and simulated performance gain of TW-1bSF over TW-SDF versus the
error-correcting capability $t$ when using BCH codes with different codeword length $n$.
``asym-Gain'' results are derived from
\eqref{eq:asym-gain}, while ``sim-EbN0 5dB'' results are the values of $P^{(\vx_1)}_{\text{TW-SDF}}/P^{(\vx_1)}_{\text{TW-1bSF}}$ obtained
by simulation.}
\end{figure}

\begin{figure}
\center\includegraphics[width=5in]{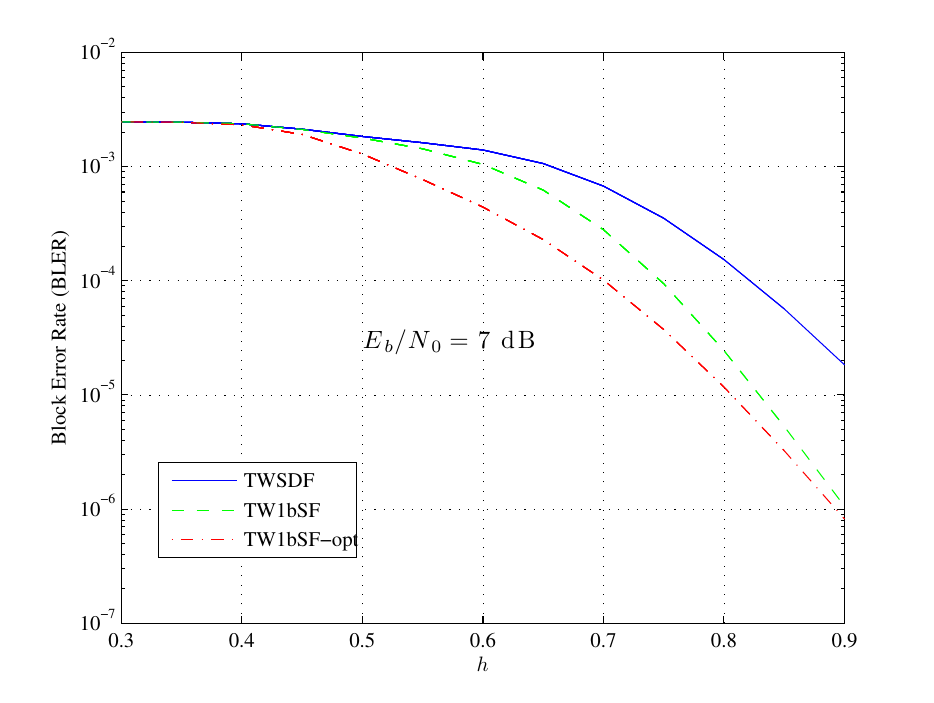}
\centerline{(a)}
\center\includegraphics[width=5in]{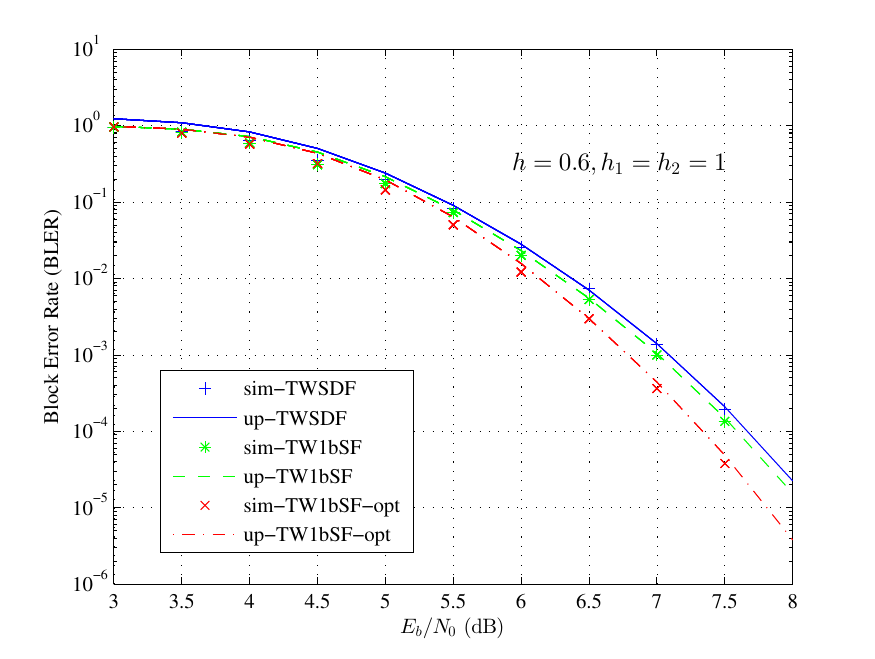}
\centerline{(b)}%\label{fig:rel-comparison-vs-SNR}
\caption{(a) Comparison of reliability value designs. The (127,113) BCH code
is employed. ``TW1bSF-opt'' stands for the results derived using the reliability value design  \eqref{eq:rel-new}, whereas ``TW1bSF'' stands for the results derived using the reliability value design  \eqref{eq:rel-old}; (b) BLER performance versus the SNR for $h=0.6$ and $h_1=h_2=1$. }
\label{fig:rel-comparison}
\end{figure}

%
%\begin{figure}
%\center\includegraphics[scale=0.7]{blerBCHN127K113}
%%##########     figure 5    ########
%#############################
%\begin{figure}
%\center\includegraphics[scale=0.7]{blerBCHN127K113}
%%\centerline{(a)}
%%\center\includegraphics[scale=0.7]{bitBCHN127K113}
%%\centerline{(b)}
%\caption{\label{fig:BCH-BER}Performance of SDF and 1bSF for the two-way network using (127,113) BCH code.}
%\end{figure}

 Fig.~\ref{fig:codelength} displays the upper bound curves as a function of the SNR to illustrate the impact of the codeword length and the error-correcting capability of the code on the performance gain of TW-1bSF over TW-SDF.
By comparing the performance when the (7,4) Hamming and the (127,113) Hamming code
are used, it can be seen that for fixed error-correcting capability $t$, the performance gain increases with the codeword length $n$. However, the comparison between
the (127,120) Hamming code and the (127,113) BCH code shows that, for fixed codeword
length $n$, higher error-correcting capability $t$ may reduce the performance gain.
For  BCH codes, Fig.~\ref{fig:asym-gain}
quantitatively depicts the asymptotic performance gain \eqref{eq:asym-gain} as a function
of the codeword length $n$, and the error-correcting capability $t$. It also presents the simulated performance gain obtained when $E_b/N_0 = 5$~dB.
In the figure, the asymptotic performance gain and the simulated performance gain demonstrate similar functional trends over the
error-correcting capability $t$, and both increase with the codeword length
$n$. However, they are not monotonically related with $t$.
%This figure would be helpful for system designers to choose appropriate BCH code
%to meet the tradeoff between the performance and the performance gain.

% ############   figure 6 & 7   ###########
%###################################
%\begin{figure}
%\center\includegraphics[scale=0.7]{diffCodesBler}
%%\centerline{(a)}
%%\center\includegraphics[scale=0.7]{diffCodesBer}
%%\centerline{(b)}
%\caption{\label{fig:codelength}Performance comparison of SDF and 1bSF for different
%codes in terms of (a) block error rate (BLER) and (b) bit error rate (BER).}
%\end{figure}
%
%
%%
%\begin{figure}
%\center\includegraphics[scale=0.7]{asymGainVStver01}
%\caption{\label{fig:asym-gain}Asymptotic performance gain versus
%error-correcting capability $t$ for a group of codeword length $n$ when BCH code is applied.}
%\end{figure}

Finally, Fig.~\ref{fig:rel-comparison} verifies the conclusion in Remark \ref{rem:rel-design} that the reliability value design
of \eqref{eq:rel-new} is better than \eqref{eq:rel-old}.
Fig.~\ref{fig:rel-comparison}(a) shows the BLER performance for the two choices of the
reliability value as a function of the channel gain of the direct link $h$, for fixed $E_b/N_0 = 7$dB. When $h$
is small, which can be interpreted as a setting with no direct link, the performance of TW-SDF
and TW-1bSF is the same for both reliability designs. The reliability
value \eqref{eq:rel-new} starts to show performance gains over TW-SDF
when $h=0.4$, and the performance gains keep increasing as $h$ grows further. By contrast, when the reliability value \eqref{eq:rel-old} is used, the performance is not superior to SDF until $h=0.5$, where Condition (2) begins to hold.
The performance improvement of design \eqref{eq:rel-new} over design \eqref{eq:rel-old} is
appealing between $h=0.6$ and $h=0.7$. Furthermore, Fig.~\ref{fig:rel-comparison}(b) shows the BLER performance of the two reliability value designs as a function of the SNR, for $h=0.6$. It can be seen that, in the BLER range from $10^{-2}$ to $10^{-4}$, the TW-1bSF protocol based on design \eqref{eq:rel-new} achieves a 0.6 dB gain over TW-SDF, which is much better than the negligible gain when design \eqref{eq:rel-old} is employed.

% ###############  figure 8 & 9  ###############
%########################################
%\begin{figure}
%\center\includegraphics[scale=0.7]{SNR7varh}
%\centerline{(a)}\label{fig:rel-comparison}
%\center\includegraphics[scale=0.7]{h06relcomp01}
%\centerline{(b)}\label{fig:rel-comparison-vs-SNR}
%\caption{(a), Comparison of reliability designs. (127,113) BCH code
%is employed. ``1bSF-rel'' stands for the result derived by using the reliability value design  \eqref{eq:rel-new} while ``1bSF'' stands for the result derived by using the reliability value design  \eqref{eq:rel-old}; (b), BLER performance versus SNR given $h=0.6$ and $h_1=h_2=1$. }
%\end{figure}

\section{Conclusions}

In this paper, we proposed an energy-efficient two-way one-bit soft forwarding (TW-1bSF) protocol to improve the performance of the TW-SDF protocol in DNC-based two-way relay networks. The idea is to always forward the network-coded packet during the third time slot even when it contains erroneous bits, instead of discarding it by the relay as in the TW-SDF protocol. The key ingredient of TW-1bSF is to also assign a reliability value to the sources. This value is used to weigh the erroneous network-coded packet, and improve the BLER. We prove theoretically, by carefully designing the reliability parameter, that the TW-1bSF protocol benefits from the superior performance of the soft relaying approach while preserving the simplicity of the TW-SDF protocol. Moreover, the bandwidth requirements are negligible, as only a one-bit indicator is sent along with the network-coded packet. Both theoretic and simulation results show that TW-1bSF attains a 0.6 dB gain over TW-SDF in practical settings. We also derive tight upper bounds on the BLER of the TW-1bSF and the TW-SDF protocols when block codes are used and hard decoding is applied at the receiving ends, and we verify the bounds by simulation.
Further analysis shows that the asymptotic performance gain of the TW-1bSF protocol over the TW-SDF protocol grows with the codeword length. This suggests that the TW-1bSF protocol can attain an impressive
performance improvement compared to TW-SDF, especially when codes with long lengthes are used. In future work, we intend to study the performance of TW-1bSF in two-way relaying networks with multiple relays and slow fading.
%From the derivation of the BLER upper bound, the design of effective reliability value is further discussed.

\appendices
\section{Proof of Theorem~\ref{thm:ineq-noe-CaseI}}\label{app:proof:theorem:caseI}

As shown in Fig.~\ref{fig:case1-block}, suppose that among the
$i$ erroneous bits of $\vy_{12}$ there are $m$ bits overlapping with the erroneous bits of $\tilde{\vy}_{r2}$. Then $m\leq t$ must hold so that $d(\vl_d)\leq t$ be possible.
Note that Fig.~\ref{fig:case1-block} only shows a special case of consecutive erroneous bits, but the following proof is general. Let $(\tilde{\vy}_{r2},\vy_{12})_{(k,m,i)}$ denote the event that $\tilde{\vy}_{r2}$ and $\vy_{12}$ have $k$ and $i$ erroneous channel bits, respectively, $m$ erroneous bits of which overlap. Given the signal $\tilde{\vy}_{r2}$ of weight $k$,
the number of such $\vy_{12}$ that correspond to the pair $(m,i)$ is
$\tbinom{k}{m}\tbinom{n-k}{i-m}$. Therefore,
\begin{IEEEeqnarray}{rCl}\label{eq:CaseI-Px}
 &&P(d(\vl_d)\leq t,d(\tilde{\vy}_{r2}) \nonumber \\
 &=&k,d(\vy_{12})\geq t+1) =\sum_{m=\max(0,t+1-(n-k))}^{t}\sum_{i=t+1}^{m+n-k}\binom{k}{m}\binom{n-k}{i-m}  P(d(\vl_d)\leq t,(\tilde{\vy}_{r2},\vy_{12})_{(k,m,i)}).
\end{IEEEeqnarray}
%
%\begin{figure}[h]
%\center\includegraphics[scale=0.55]{caseI}
%\caption{\label{fig:case1-block} The alignment between $\tilde{\vy}_{r2}$ and $\vy_{12}$ for decoding $\vl_d=4h\sqrt{E}\vy_{12}+4h_2\sqrt{E_r}\tilde{\vy}_{r2}$ in Case I. The erroneous channel bits are represented by shadow box, while the correct channel bits are represented by blank box.}
%\end{figure}

%
To derive $P(d(\vl_d)\leq t,(\tilde{\vy}_{r2},\vy_{12})_{(k,m,i)})$ we first observe that $d(\vl_d)\leq t$ occurs if there are $g\leq t-m$ erroneous channel bits among the bits $\{\tau:\tilde{y}_{r2,\tau}<0,y_{12,\tau}>0\}$ and $\{\tau:\tilde{y}_{r2,\tau}>0,y_{12,\tau}<0\}$. Suppose that $j$ of the $g$ erroneous bits are within the set $\{\tau:\tilde{y}_{r2,\tau}<0,y_{12,\tau}>0\}$, and that the remaining $g-j$ erroneous bits are within the set $\{\tau:\tilde{y}_{r2,\tau}>0,y_{12,\tau}<0\}$, as shown in Fig.~\ref{fig:case1-block}. For the $\tau$-th channel bit, let
%%
%$p_{\co}=P(\tilde{y}_{r2,\tau}>0,y_{12,\tau}>0)= \left(1-p_{r2}\right)\left(1-p_{12}\right)$,
%%
%and $q_{\co}=P(\tilde{y}_{r2,\tau}<0,y_{12,\tau}<0)= p_{12}p_{r2}$,
%%
%and $\dot{p}_{\co}=P(\ell_{d,\tau}<0,\tilde{y}_{r2,\tau}<0,y_{12,\tau}>0)=\Pr(h_{2}\sqrt{E_r}\tilde{y}_{r2,\tau}+h\sqrt{E}y_{12,\tau}<0,y_{12,\tau}>0)$,
%%
%and $\dot{q}_{\co}=P(\ell_{d,\tau}>0,\tilde{y}_{r2,\tau}<0,y_{12,\tau}>0)=p_{r2}(1-p_{12})-\dot{p}_{\co}$,
%%
%and $\ddot{p}_{\co}=P(\ell_{d,\tau}<0,\tilde{y}_{r2,\tau}>0,y_{12,\tau}<0)=\Pr(h_{2}\sqrt{E_r}\tilde{y}_{r2,\tau}+h\sqrt{E}y_{12,\tau}<0,\tilde{y}_{r2,\tau}>0)$,
%%
%and $\ddot{q}_{\co}=P(\ell_{d,\tau}>0,\tilde{y}_{r2,\tau}>0,y_{12,\tau}<0)=p_{12}(1-p_{r2})-\ddot{p}_{\co}$,
\begin{IEEEeqnarray}{rCl}
p_{\co}&=&P(\tilde{y}_{r2,\tau}>0,y_{12,\tau}>0)= \left(1-p_{r2}\right)\left(1-p_{12}\right) \IEEEyessubnumber \\
q_{\co}&=&P(\tilde{y}_{r2,\tau}<0,y_{12,\tau}<0)= p_{12}p_{r2} \IEEEyessubnumber \\
\dot{p}_{\co}&=&P(\ell_{d,\tau}<0,\tilde{y}_{r2,\tau}<0,y_{12,\tau}>0) =\Pr(h_{2}\sqrt{E_r}\tilde{y}_{r2,\tau}+h\sqrt{E}y_{12,\tau}<0,y_{12,\tau}>0) \IEEEyessubnumber \\
 \dot{q}_{\co}&=&P(\ell_{d,\tau}>0,\tilde{y}_{r2,\tau}<0,y_{12,\tau}>0) =p_{r2}(1-p_{12})-\dot{p}_{\co} \IEEEyessubnumber \\
 \ddot{p}_{\co}&=&P(\ell_{d,\tau}<0,\tilde{y}_{r2,\tau}>0,y_{12,\tau}<0) =\Pr(h_{2}\sqrt{E_r}\tilde{y}_{r2,\tau}+h\sqrt{E}y_{12,\tau}<0,\tilde{y}_{r2,\tau}>0)  \IEEEyessubnumber \\
 \ddot{q}_{\co}&=&P(\ell_{d,\tau}>0,\tilde{y}_{r2,\tau}>0,y_{12,\tau}<0)  =p_{12}(1-p_{r2})-\ddot{p}_{\co}, \IEEEyessubnumber
\end{IEEEeqnarray}
where $\ell_{d,\tau}=4h\sqrt{E}y_{12,\tau}+4h_2\sqrt{E_r}\tilde{y}_{r2,\tau}$ with $y_{12,\tau}\sim \mathcal{N}(h\sqrt{E},\tfrac{1}{2})$ and $\tilde{y}_{r2,\tau}\sim \mathcal{N}(h_2\sqrt{E_r},\tfrac{1}{2})$. By counting all combinations of $g$ and $j$, we get
\begin{eqnarray}\label{eq:caseI-Pri}
 && P(d(\vl_d)\leq t,(\tilde{\vy}_{r2},\vy_{12})_{(k,m,i)}) \cr
& = & p_{\co}^{n-k-(i-m)}q_{\co}^{m}\sum_{g=0}^{t-m}\sum_{j=0}^{g}\binom{k-m}{j}
\dot{p}_{\co}^{j}\dot{q}_{\co}^{k-m-j}  \binom{i-m}{g-j}\ddot{p}_{\co}^{g-j}\ddot{q}_{\co}^{i-m-(g-j)}.
\end{eqnarray}
Finally, substituting \eqref{eq:caseI-Pri} into \eqref{eq:CaseI-Px}, and then into \eqref{eq:ineq-noe-CaseI}, we get \eqref{eq:lb-noe-CaseI}, which proves the theorem.

\section{Proof of Theorem~\ref{thm:BLER_SDF_upb}} \label{app:proof:thm:BLER_SDF_upb}
We first prove two inequalities. The first one is
\begin{align}
\Theta_{\co}\cup \Theta_{\cc} \cup \Theta_{\ccc}
&=[\e_{12}\cap \e^c_{1r}\cap \e^c_{2r}\cap \e_{r2}] \cup [\e_{12}\cap(\e_{1r}\cup \e_{2r})]  =   \e_{12}\cap (\e^c_{1r}\cap \e^c_{2r}\cap \e^c_{r2})^c \cr
& \subseteq  \grave{\e}_{12}\cap (\grave{\e}^c_{1r}\cap \grave{\e}^c_{2r}\cap \grave{\e}^c_{r2})^c  =  \grave{\Theta}_{\co}\cup \grave{\Theta}_{\cc} \cup \grave{\Theta}_{\ccc},
\end{align}
which is derived using the facts $\e_{12}\subseteq \grave{\e}_{12}$, $\grave{\e}^c_{1r}\subseteq \e^c_{1r}$,
$\grave{\e}^c_{2r}\subseteq \e^c_{2r}$ and $\grave{\e}^c_{r2}\subseteq \e^c_{r2}$. The other is
\begin{equation}
\Theta_{\cc}\cup \Theta_{\ccc}  =   \e_{12}\cap(\e_{1r}\cup \e_{2r})  \subseteq  \grave{\e}_{12}\cap(\grave{\e}_{1r}\cup \grave{\e}_{2r})  =   \grave{\Theta}_{\cc} \cup \grave{\Theta}_{\ccc}.
\end{equation}
We then decompose $\Theta_\co$ as
\begin{equation}
\Theta_\co  =  \underset{\Theta_{\co.a}}{ \underbrace{\e_{12}\cap \grave{\e}^c_{1r}\cap \grave{\e}^c_{2r}\cap \e_{r2}} }
 + \underset{\Theta_{\co.b}}{\underbrace{\e_{12}\cap (\e^c_{1r}\diagdown \grave{\e}^c_{1r})\cap
(\e^c_{2r}\diagdown \grave{\e}^c_{2r})\cap \e_{r2}}}.
\end{equation}
Because $\Theta_{\co.b}\cup \Theta_\cc \cup \Theta_{\ccc} \subseteq \Theta_{\co}\cup \Theta_{\cc} \cup \Theta_{\ccc}
\subseteq \grave{\Theta}_{\co}\cup \grave{\Theta}_{\cc} \cup \grave{\Theta}_{\ccc}$ and $(\Theta_{\co.b}\cup \Theta_\cc \cup \Theta_\ccc) \cap \grave{\Theta}_{\co} = \emptyset$, we conclude that $\Theta_{\co.b}\cup \Theta_\cc \cup \Theta_\ccc \subseteq \grave{\Theta}_{\cc} \cup \grave{\Theta}_{\ccc}$, and thus
\begin{equation}\label{eq:gTcc_gTccc}
P(\grave{\Theta}_{\cc}\cup \grave{\Theta}_{\ccc})  \geq   P(\Theta_{\co.b}\cup \Theta_\cc \cup \Theta_\ccc)
=P(\Theta_{\co.b}) + P(\Theta_\cc \cup \Theta_\ccc).
\end{equation}
Combining \eqref{eq:gTcc_gTccc} with $\Theta_{\co.a}\subseteq \grave{\Theta}_\co$, we have
\begin{align}
P_{\text{TW-SDF}}^{(\vx_1)} &=  P(\e,\Theta_{\co})+P(\Theta_{\cc}\cup \Theta_{\ccc})  =  P(\e,\Theta_{\co.a})+ P(\e,\Theta_{\co.b}) +P(\Theta_{\cc}\cup \Theta_{\ccc}) \cr
& \leq  P(\e,\Theta_{\co.a})+ P(\Theta_{\co.b}) +P(\Theta_{\cc}\cup \Theta_{\ccc})  \leq  P(\e,\grave{\Theta}_{\co}) + P(\grave{\Theta}_{\cc} \cup \grave{\Theta}_{\ccc})  \cr
& \leq  P(\grave{\e},\grave{\Theta}_{\co}) + P(\grave{\Theta}_{\cc} \cup \grave{\Theta}_{\ccc}),
\end{align}
which proves the theorem.

\section{Proof of Theorem~\ref{thm:barP_caseII.1}} \label{app:proof:thm:barP_caseII.1}

\begin{figure}
\center\includegraphics[scale=0.4]{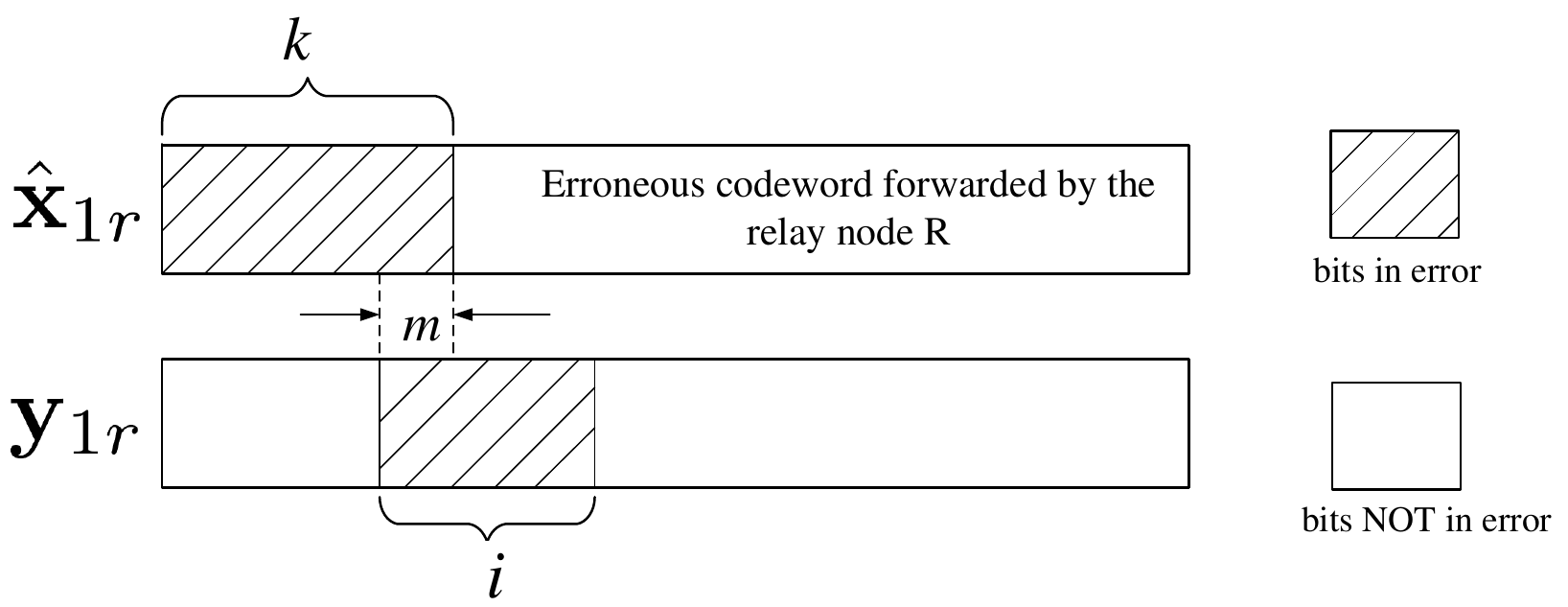}
\caption{\label{fig:SF-caseII.1}Error position alignment between the received
signals from the $R$-$S2$ link and the $S1$-$S2$ link for Case II.a. }
\end{figure}

In the subcase $\grave{\Theta}_{\cc a}$, note that $\D(\tilde{\vy}_{r2})=\hat{\vx}_{1r}$ because $\grave{\e}^c_{r2}  \Rightarrow \D(\tilde{\vy}_{r2}) = \hat{\vx}_{1r} \circ \hat{\vx}_{2r} \circ \vx_2$ and $\grave{\e}^c_{2r} \Rightarrow  \hat{\vx}_{2r} = \vx_2 = \mathbf{1}$. Using this fact and the inequality $P({\e}^c)\geq P\left(d(\vl_d)\leq t\right)$ with $\vl_d = 4h\sqrt{E}\vy_{12} + \hat{\vx}_{1r}\mathcal{L}$, we have
\begin{eqnarray}\label{eq:ineq2-CaseII.1}
&& \makebox[-1cm]{} P\left({\e}^c,d(\vy_{12})\geq t+1 \mid d(\hat{\vx}_{1r})=k,\grave{\e}^c_{2r},\grave{\e}^c_{r2}\right) \cr
&\geq&  P( d(\vl_d)\leq t,d(\vy_{12})\geq t+1 \mid d(\hat{\vx}_{1r})=k,\mathfrak{D}(\tilde{\vy}_{r2})=\hat{\vx}_{1r} ).
\end{eqnarray}

Now, given $d(\hat{\vx}_{1r})=k$ and $d(\vy_{12}) = i$ as shown in Fig.~\ref{fig:SF-caseII.1}, the inequality $d(\vl_d)\leq t$ holds only if the cardinality of $\{\tau:\hat{x}_{1r,\tau}=-1, y_{12,\tau}<0\}$ does not exceed $t$, i.e., $m\leq t$, and among the bits $\{\tau:\hat{x}_{1r,\tau}=-1,y_{12,\tau}>0\}$ and $\{\tau:\hat{x}_{1r,\tau}=1,y_{12,\tau}<0\}$ there are only $g\leq t-m$ erroneous channel bits. Suppose that $j$ erroneous bits are within the set $\{\tau:\hat{x}_{1r,\tau}=-1,y_{12,\tau}>0\}$, and that the remaining $g-j$ erroneous bits are within the set $\{\tau:\hat{x}_{1r,\tau}=1,y_{12,\tau}<0\}$.
For the $\tau$-th channel bit, let
\begin{subequations}
\begin{align}
p_{\cc}&=P(y_{12,\tau}>0)=1-p_{12}, \\
q_{\cc}&=P(y_{12,\tau}<0)=p_{12}, \\
\dot{p}_{\cc}&=P(\ell_{d,\tau}<0,y_{12,\tau}>0\mid \hat{x}_{1r,\tau}=-1) =\Pr(4h\sqrt{E}y_{12,\tau}-\mathcal{L}<0,y_{12,\tau}>0),  \\
\dot{q}_{\cc}&=P(\ell_{d,\tau}>0,y_{12,\tau}>0\mid \hat{x}_{1r,\tau}=-1) =\Pr(4h\sqrt{E}y_{12,\tau}-\mathcal{L}>0,y_{12,\tau}>0), \\
\ddot{p}_{\cc}&=P(\ell_{d,\tau}<0,y_{12,\tau}<0\mid \hat{x}_{1r,\tau}=1) =\Pr(4h\sqrt{E}y_{12,\tau}+\mathcal{L}<0,y_{12,\tau}<0), \\
\ddot{q}_{\cc}&=P(\ell_{d,\tau}>0,y_{12,\tau}<0\mid \hat{x}_{1r,\tau}=1) =\Pr(4h\sqrt{E}y_{12,\tau}+\mathcal{L}>0,y_{12,\tau}<0).
\end{align}
\end{subequations}
Then by summing the probability of correct detection
$p_{\cc}^{n-k-(i-m)}q_{\cc}^{m}\dot{p}_{\cc}^{j}\dot{q}_{\cc}^{k-m-j}\ddot{p}_{\cc}^{g-j}\ddot{q}_{\cc}^{i-m-(g-j)}$
over the combinations of the $5$-tuples $(k,i,m,g,j)$, we have
\begin{eqnarray}\label{eq:noe_yrk}
&& \makebox[-.7cm]{}P(d(\vl_d)\leq t,d(\vy_{12})\geq (t+1)\mid d(\hat{\vx}_{1r})=k,\mathfrak{D}(\tilde{\vy}_{r2})=\hat{\vx}_{1r})
\nonumber \\
& = & \sum_{m=\max(0,t+1-(n-k))}^{t}\sum_{i=t+1}^{m+(n-k)}\binom{k}{m}\binom{n-k}{i-m} \nonumber \\
&& \makebox[0.3cm]{}\times p_{\cc}^{n-k-(i-m)}p_{12}^{m}\sum_{g=0}^{t-m}\sum_{j=0}^{g}\binom{k-m}{j}
\dot{p}_{\cc}^{j}\dot{q}_{\cc}^{k-m-j}\binom{i-m}{g-j}\ddot{p}_{\cc}^{g-j}\ddot{q}_{\cc}^{i-m-(g-j)}.
\end{eqnarray}
Finally, substituting  \eqref{eq:noe_yrk} into \eqref{eq:ineq2-CaseII.1} and then into \eqref{eq:ineq-CaseII.1}, we obtain \eqref{eq:lb-CaseII.1}, and thus prove the theorem.

\section{Proof of Lemma~\ref{lemma:asym-ratio}} \label{app:proof:lem:asym-ratio}

Let $ 0 < \epsilon < \min\{\mu_1,\mu_2\}$ be a small positive value. We can lower bound $\Lambda_1$ as
\begin{align}
\Lambda_1
  = & \left(\int_{-(\mu_1-\epsilon)\sqrt{E}}^{0}+\int_{-\infty}^{-(\mu_1-\epsilon)\sqrt{E}}\right) \left(1-Q((\mu_1\sqrt{E}+z_2)\sqrt{2})\right)f(z_2)dz_2\nonumber \\
  \geq & \left(1-Q(\epsilon\sqrt{2E})\right)\left(Q(\mu_2\sqrt{2E})-Q((\mu_2+\mu_1-\epsilon)\sqrt{2E})\right),\nonumber
\end{align}
which is derived by lower bounding the second integral by $0$, and the first integral using the fact that
\begin{equation}
\left(1-Q((\mu_1\sqrt{E}+z_2)\sqrt{2})\right)\mid _{z_2\in[0,-(\mu_1-\epsilon)\sqrt{E}]}
\geq  \left(1-Q((\mu_1\sqrt{E}+z_2)\sqrt{2})\right)\mid _{z_2=-(\mu_1-\epsilon)\sqrt{E}}. \nonumber
\end{equation}
On the other hand,  we can upper bound $\Lambda_2$ as
\begin{align}
\Lambda_2
 & =  \left(\int_{-(\mu_1-\epsilon)\sqrt{E}}^{0}+\int_{-\infty}^{-(\mu_1-\epsilon)\sqrt{E}}\right)\left(Q((\mu_1\sqrt{E}+z_2)\sqrt{2})-Q(\mu_1\sqrt{2E})\right)f(z_2)dz_2\nonumber \\
 & \leq  \left(Q(\epsilon\sqrt{2E})-Q(\mu_1\sqrt{2E})\right)\left(Q(\mu_2\sqrt{2E})-Q((\mu_2+\mu_1-\epsilon)\sqrt{2E})\right) + Q(-\mu_1\sqrt{2E}) Q((\mu_2+\mu_1-\epsilon)\sqrt{2E}). \nonumber
\end{align}
Thus when $E\rightarrow\infty$, we have $Q(-\mu_1\sqrt{2E})\rightarrow 1$, $Q(\epsilon\sqrt{2E})/ Q(\mu_1\sqrt{2E})\rightarrow \infty$ and $Q(\mu_2\sqrt{2E})/Q((\mu_2+\mu_1-\epsilon)\sqrt{2E})\rightarrow \infty$, and finally
\begin{equation}
\frac{\Lambda_1}{\Lambda_2}\geq\frac{(1-Q(\epsilon\sqrt{2E}))
Q(\mu_2\sqrt{2E})}{Q(\epsilon\sqrt{2E})Q(\mu_2\sqrt{2E})+Q((\mu_2+\mu_1-\epsilon)\sqrt{2E})}\rightarrow\infty, \nonumber
\end{equation}
which proves the lemma.

\bibliographystyle{IEEEtran}

%\bibliographystyle{plain}
%\bibliography{IEEEabrv,1bitSFTWRC}
%\section*{Reference}

%%%%%   figures  %%%%%
%%%%%%%%%%%%%%%

%\section*{Appendix A: Proof of Theorem \ref{thm:barP_caseI}} \label{app:proof-T1}

%

%\section*{Appendix A: Proof of Theorem \ref{thm:barP_caseII.1}} \label{app:proof-T2}

%

\end{document}